\documentclass[twocolumn,10pt]{IEEEtran}

\usepackage{multicol,bbm}
\usepackage[normalem]{ulem}
\usepackage{url}
\usepackage{hyperref}
\hypersetup{colorlinks,urlcolor=blue}

\hypersetup{colorlinks=false,pdfborder=000}
\usepackage{booktabs}

\input{def.tex}
\usepackage{dsfont}
\usepackage{minibox}
\DeclareSymbolFont{matha}{OML}{txmi}{m}{it}
\DeclareMathSymbol{\varv}{\mathord}{matha}{118}
\usepackage{eurosym}
\usepackage{hhline,physics}

\usepackage{multicol}
\usepackage{algorithm}
\usepackage{graphicx}
\usepackage{textcomp}
\usepackage{caption,pifont}
\usepackage[multiple]{footmisc}

\usepackage{algorithmic}

\newcommand{\trans}{^{\mathsf{T}}}
\newcommand{\herm}{^{\H}}
\makeatletter

\IEEEoverridecommandlockouts

\begin{document}
	\title{Ergodic Mutual Information and Outage Probability for SIM-Assisted Holographic MIMO Communications} 
	\author{Anastasios Papazafeiropoulos, Pandelis Kourtessis,  Dimitra I. Kaklamani, 			Iakovos S. Venieris \thanks{A. Papazafeiropoulos is with the Communications and Intelligent Systems Research Group, University of Hertfordshire, Hatfield AL10 9AB, U. K.  P. Kourtessis is with the Communications and Intelligent Systems Research Group, University of Hertfordshire, Hatfield AL10 9AB, U. K.   Dimitra I. Kaklamani is with the Microwave and Fiber Optics Laboratory, and Iakovos S. Venieris is  with the Intelligent Communications and Broadband Networks Laboratory, School of Electrical and Computer Engineering, National Technical University of Athens, Zografou, 15780 Athens,	Greece.	
			Corresponding author's email: tapapazaf@gmail.com.}}
	\maketitle\vspace{-1.7cm}
	\begin{abstract}	
		Stacked intelligent metasurface (SIM) is a promising enabler for next-generation high-capacity networks that exhibit better performance compared to its single-layer counterpart by means of just wave propagation. However, the study of ergodic mutual information (EMI) and outage probability for SIM-assisted multiple-input-multiple-output (MIMO) systems is not available in the literature. To this end, we obtain the distribution of the MI by using large random matrix theory (RMT) tools. Next, we derive a tight closed-form expression for the outage probability  based on statistical channel state information (CSI). Moreover, we apply the gradient descent method for the  minimization of the outage probability. 
		Simulation results verify  the analytical results and provide fundamental insights such as the performance enhancements compared to conventional MIMO systems and the single-layer counterpart. Notably the proposed optimization algorithm is faster than the alternating optimization (AO) benchmark by saving significant overhead.	
	\end{abstract}
	\begin{keywords}
		Reconfigurable intelligent surface 	(RIS), stacked intelligent metasurfaces (SIM),  gradient projection,  6G networks.
	\end{keywords}
	
	\section{Introduction}
	
	A promising technology to improve the performance of sixth-generation (6G) communication systems while offering the potential to shape the propagation environment is 	the reconfigurable intelligent surface (RIS) by means of the configuration of nearly passive reflective elements \cite{DiRenzo2020,Wu2020,Papazafeiropoulos2021,Papazafeiropoulos2023c}. A RIS is implemented by a large number of  elements having a low cost that can induce phase shifts to the impinging waves with diverse objectives towards better connectivity \cite{DiRenzo2020}. 
	
	Plenty of works have studied the performance analysis and design of RIS-assisted systems \cite{Wu2019,Bjoernson2019b,Yang2020b,Zhao2020,Papazafeiropoulos2021,Mu2021,Papazafeiropoulos2023,Papazafeiropoulos2023a, Papazafeiropoulos2023b}. In particular, the capacity and the outage probability are two significant performance metrics. On this ground, the ergodic mutual information (EMI) has been studied for RIS-assisted multiple-input-multiple-output (MIMO) systems over Rician channels through a tight and closed-form approximation \cite{Zhang2021a}. Regarding the outage probability, in the case of RIS-assisted single-input single-output (SISO) systems, it has been evaluated in several previous works, e.g., \cite{Salhab2021}. Similarly, in the case of multiple-input single-output (MISO) systems, an example is the derivation of the outage probability and its optimal expression with maximum-ratio transmission (MRT) in \cite{Guo2020}. 
	Previous methods concerning SISO and MISO systems cannot be applied to MIMO systems since the spectral distribution of the product of two random matrices is a challenging problem \cite{Zheng2016}. Note that, the outage probability of RIS-assisted MIMO has been investigated only in \cite{Shi2021} and \cite{Zhang2022}. In the former, the Mellin transform was exploited to obtain the outage probability over Rayleigh fading channels with channel correlation only on one side of the transceiver. In the latter, the statistics of the MI for RIS-assisted MIMO systems was derived and it was used to study the outage probability.
	
	Despite the suggestion of RIS for various communication scenarios due to its numerous advantages, the majority of existing works have focused on single-layer metasurface structures, which present limitations on the  beam management \cite{Guo2020a}. Also, the single-layer configuration of conventional single-layer RIS and hardware limitations cannot suppress inter-user interference. These gaps lead \textit{An et al}. to the concept of   stacked intelligent metasurface (SIM) \cite{An2023}, which comes with  remarkable advantages compared to a single-layer conventional RIS. In particular, a SIM-based transceiver for point-to-point MIMO communication was proposed, where two SIMs were positioned at the transmitter and the receiver with the electromagnetic waves (EM) propagating through them without the use of any digital hardware. Each SIM was optimized by means of an alternating optimization (AO) approach under instantaneous CSI conditions. On the contrary, in \cite{Papazafeiropoulos2024a}, a concurrent optimization of the phase shifts of both SIMs took place based on  hybrid digital-wave design. In this direction, the achievable rate of SIM-assisted systems was studied for multi-user MISO systems in \cite{Papazafeiropoulos2024,Papazafeiropoulos2024b} while  the  near field beamforming was investigated in    \cite{Papazafeiropoulos2024c}.
	
	\textit{Contributions}: The above observations were the motivation for this paper, where we derive the EMI and \textcolor{black}{outage} probability for large SIM-assisted holographic MIMO (HMIMO) systems based on statistical CSI. Moreover, we propose a gradient descent algorithm to minimise the outage probability by optimising simultaneously the phase shifts of the transceiver. The proposed approach is more advantageous compared to an instantaneous-based approach, which is crucial for the implementation of large SIM-assisted MIMO systems.  \textcolor{black}{Compared to \cite{Zhang2022}, our manuscript investigates the EMI and outage probability of SIM-assisted MIMO systems and highlights their performance differences compared to conventional RIS-assisted architectures. The proposed SIM-based system exhibits substantial distinctions from the RIS-based model in \cite{Zhang2022}, both in terms of architecture and mathematical modeling. Unlike RIS, which is typically deployed in the intermediate space between the transmitter and receiver, SIMs are positioned at both the transmitter and the receiver ends. Furthermore, our system incorporates two SIMs, whereas \cite{Zhang2022} considers a single RIS. From a structural perspective, SIMs are multilayer configurations that result in a non-diagonal overall transfer matrix, in contrast to RISs, which are typically single-layer structures with a diagonal phase shift matrix.} 	
	
\textcolor{black}{	Contary to \cite{An2023,Papazafeiropoulos2024a}, this work adopts a statistical CSI framework, where the optimization of phase shifts is performed based solely on long-term channel statistics (e.g., correlation matrices). This reduces signaling overhead and enhances scalability, making the method suitable for practical implementation. Furthermore, we propose a projected gradient descent algorithm that jointly optimizes both the transmitter and receiver SIMs, avoiding the limitations of sequential AO updates used in \cite{An2023}. To provide deeper theoretical insights, we develop a large system analysis using random matrix theory (RMT), which yields closed-form expressions for the EMI, outage probability, and finite-SNR diversity-multiplexing tradeoff DMT, which are performance metrics not addressed in prior SIM works. Our results establish fundamental trade-offs and demonstrate the effectiveness of SIM optimization under realistic statistical CSI assumptions.}
	
	The key points of this paper are summarized below. 
	\begin{itemize}
		\item 	We obtain the EMI for SIM-assisted MIMO systems over correlated channels based on statistical CSI. Next, we derive an approximation of the outage probability. Notably, the validation of the proposed method is presented by numerical results. 
		\item	We optimise the outage probability by means of a gradient descent algorithm optimising the phase shifts of the two SIMs simultaneously while assuming only statistical CSI, which saves significant overhead. 
		\item	We obtain the finite signal-to-noise ratio (SNR) \textcolor{black}{diversity-multiplexing tradeoff (DMT)} for large SIM-assisted MIMO systems, which depends on the mean and the variance of the MI. 
		\item	We present analytical and simulation results that coincide and shed light on the impact of the system parameters on the EMI and the outage probability.
	\end{itemize}
	
	\textit{Paper Outline}: Section~\ref{System} presents the system and channel models of a SIM-assisted MIMO system.   Section~\ref{mutualI} provides the EMI.  Section~\ref{outageProb}  presents the outage probability together with its optimization. Next, we provide the finite-SNR DMT.  In Section~\ref{Numerical}, we depict and elaborate on the numerical results,  and Section~\ref{Conclusion} concludes the paper.	
	\paragraph*{Notations}
	Bold uppercase and lowercase letters are used to denote matrices and vectors respectively. $\bX \trans$, $\bX \herm$,  and $\tr (\bX)$ denote the (ordinary) transpose, conjugate transpose, and trace of the matrix $\bX$. A square diagonal matrix whose main diagonal consists of the elements of vector $\bx$ is represented by $\diag (\bx)$.  The complex-valued gradient of $f(\cdot)$ with respect to (w.r.t.) $\bX^{\ast}$ is denoted by $\nabla_{\bX} f(\cdot)$.  The vector space of all complex-valued matrices of size $M \times N$ is denoted by $\mathbb C^{M \times N}$. The amplitude and phase of a complex number $x$ are denoted by $|x|$ and $\angle x$, respectively. The statistical expectation and variance operators are denoted by using $\mathbb E \{\cdot\}$ and $\mathrm{Var}\{\cdot\}$, respectively. 	$\Xi(x)$ is the cumulative 	distribution function (CDF) of standard Gaussian distribution 	and $Q(\cdot)$ is the $Q$-function, where $Q(x) = 1-\Xi(x)$.  $ \xrightarrow[N \to \infty]{\mathcal{P}}$, $ \xrightarrow[N \to \infty]{\mathcal{D}}$, and $ \xrightarrow[N \to \infty]{}$ 	denote the 	convergence in probability, the convergence in distribution, and 	the almost sure convergence, respectively.
	
	\section{\textcolor{black}{Architecture and Modeling of SIM-Based HMIMO}}\label{System}
	\subsection{System Model}
	\textcolor{black}{	Consider a SIM-aided point-to-point holographic MIMO (HMIMO) communication system, as illustrated in Fig. \ref{Fig01}. The transmitter is equipped with $N_t$ antennas and the receiver with $N_r$ antennas. For clarity, these antennas are not depicted in the figure, which focuses solely on the SIM implementation. }
	A smart controller adjusts the phase shift of the EM waves impinging on each meta-atom of each surface. 
	
	\begin{figure}
		\begin{center}
			\includegraphics[width=0.8\linewidth]{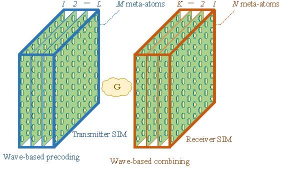}
			\caption{A SIM-assisted HMIMO system. }
			\label{Fig01}
		\end{center}
	\end{figure}
	
	Based on the SIM architecture  in \cite{An2023,Papazafeiropoulos2024a,Papazafeiropoulos2024b,Papazafeiropoulos2024}, let  $ L $ be the number of transmitter SIM layers
	and $ K $ be the number of receiver SIM layers. Also,  let $ M $ be the number of meta-atoms on each  layer at the transmitter SIM and let  $ N $ be the number of meta-atoms on each  layer at the  receiver SIM.  Let the corresponding sets be  $ \mathcal{L}=\{1,\ldots,L\} $,  $\mathcal{K}=\{1,\ldots,K\} $,  $ \mathcal{M}=\{1,\ldots,M\} $,  and $ \mathcal{N}=\{1,\ldots,N\} $, respectively.
	To this end, the phase shift by  atom  $ m $ on the  transmit  layer $ l $  is denoted by  $ \theta_{m}^{l}\in [0,2\pi), m \in \mathcal{M}, l \in \mathcal{L} $ be with $ \phi_{m}^{l} =e^{j \theta_{m}^{l}}$. The transmission coefficient matrix, associated with the $l$-th transmit layer, is given by $ \bPhi^{l}=\diag(\bphi^{l})\in \mathbb{C}^{M \times M} $, where $ \bphi^{l} =[\phi^{l}_{1}, \dots, \phi^{l}_{M}]^{\T}\in \mathbb{C}^{M \times 1}$. In a similar way,  the phase shift by meta-atom  $ n $ on the  receive metasurface layer  $ k $ is denoted by $ \xi_{n}^{k}\in [0,2\pi), n \in \mathcal{N}, k \in \mathcal{K} $  with $ \psi_{n}^{k}=e^{j \xi_{n}^{k}} $ being the corresponding transmission coefficient. The $k$-th receive  matrix is given by $ \bPsi^{k}=\diag(\bpsi^{k})\in \mathbb{C}^{M \times N} $ with $ \bpsi^{k} =[\psi^{k}_{1}, \dots, \psi^{k}_{M}]^{\T}\in \mathbb{C}^{N \times 1}$. Note that we assume constant modulus equal to $ 1 $ and continuously-adjustable phase shifts \cite{Wu2019,An2023,Papazafeiropoulos2024a}. Moreover, we rely on the design of both SIMs  described in \cite{An2023,Papazafeiropoulos2024}. Note that future work could address more realistic assumptions, e.g., the case of  coupled phase and magnitude \cite{Abeywickrama2020}.\footnote{\textcolor{black}{In this work, we adopt the commonly used unit-modulus assumption for the SIM elements, implying lossless transmission through each metasurface layer. While this simplifies the analysis and highlights the fundamental benefits of phase optimization, practical implementations may exhibit insertion losses and amplitude-phase coupling due to material and fabrication constraints. Incorporating such impairments is a promising direction for future research.}}

	\subsection{Channel Model}
	The  transmitter SIM is described as
	\begin{align}
		\bP=\bPhi^{L}\bW^{L}\cdots\bPhi^{2}\bW^{2}\bPhi^{1}\bW^{1}\in \mathbb{C}^{M \times N_{t}},\label{TransmitterSIM}
	\end{align}
	where $ \bW^{l}\in \mathbb{C}^{M \times M}, l \in \mathcal{L}/\{1\} $ is the transmission coefficient matrix between layers $ (l-1)$ and $ l $ while  $ \bW^{1} \in \mathbb{C}^{M \times  N_{t}} $ corresponds to the input  layer of the  SIM at the transmitter side. Note that the transmission coefficient between meta-atom  $ \tilde{m} $ and meta-atom  $ m $ is written according  to the  Rayleigh-Sommerfeld diffraction theory \cite{Lin2018,An2023}. 
%
	In a similar way, the  receiver SIM is characterized by
	\begin{align}
		\bD=\bU^{1}\bPsi^{1}\bU^{2}\bPsi^{2}\cdots\bU^{K}\bPsi^{K}\in \mathbb{C}^{N_{r} \times N},
	\end{align}
	where $ \bU^{k}\in \mathbb{C}^{N \times N}, k \in \mathcal{K}/\{1\} $ denotes the transmission coefficient matrix between the  receive  layers $ k $ and $ (k-1) $. Also, $ \bU^{1}\in \mathbb{C}^{N_{r} \times N}$ is the matrix from the output layer of the receiver SIM to the receive antenna array. Here, the  coefficient from  meta-atom
	$ n $ to   meta-atom $ \tilde{n} $  is written as in \cite{Lin2018,An2023} again.

	The HMIMO channel between the transmitter and receiver SIMs  is described by the $N\times M$ complex matrix \cite{Hu2022}
	\begin{align}
		\bG=\bR^{1/2}_{\mathrm{R}}\tilde{\bG}\bR^{1/2}_{\mathrm{T}},\label{channel}
	\end{align}
	where $ \bR_{\mathrm{R}}\in \mathbb{C}^{N\times N} $  and $\bR_{\mathrm{T}}\in \mathbb{C}^{M\times M} $ are spatial correlation matrices at the receiver and transmitter SIMs, respectively.\footnote{\textcolor{black}{Future work could focus on extending the current framework to support more general channel models with non-separable spatial correlation structures \cite{Zhang2025}, which more accurately represents holographic MIMO scenarios. Incorporating such models into the analysis and optimization of SIM-assisted systems will provide deeper insights into their performance under realistic electromagnetic coupling conditions.}} Specifically, under the assumptions of far-field propagation and isotropic scattering \cite{Pizzo2020,Dai2020}, the spatial correlation matrices at the   transmitter  and receiver SIMs are described by \cite{Demir2022} as
	\begin{align}
				[\bR_{\mathrm{T}}]_{m,\tilde{m}}&=\mathrm{sinc}(2 r_{m,\tilde{m}}/\lambda),  m\in \mathcal{M}, \tilde{m}\in \mathcal{M},\label{t}\\
		[\bR_{\mathrm{R}}]_{\tilde{n},n}&=\mathrm{sinc}(2 t_{\tilde{n},n}/\lambda),  \tilde{n}\in \mathcal{N}, n\in \mathcal{N},\label{c}
	\end{align}
	respectively. Moreover, $ \tilde{\bG}\sim \mathcal{CN}(\b0,\frac{\beta}{M}\Id_{N}\otimes \Id_{M})\in \mathbb{C}^{N\times M} $ expresses the independent  Rayleigh fading channel,  	Also,  $\beta $ expresses the average path loss between the two  SIMs, which is written as 
	given by \cite{Rappaport2015}
	\begin{align}
		\beta(d)=\beta(d_{0})+10 b \log_{10}\left(\frac{d}{d_{0}}\right)+X_{\delta},~d \ge d_{0},
	\end{align}
	where  $ b $ is the path loss exponent,  $ X_{\delta}\sim \mathcal{CN}(0, \delta) $ with  $ \delta $  depending on shadow fading, and $ 
	\beta(d_{0})=20 \log_{10}(4 \pi d_{0}/\lambda)~\mathrm{dB}$ is the free space path loss at a reference distance $ d_{0} $. Notably, under these settings, the total $ \bH \in \mathbb{C}^{N_{r} \times N_{t}} $  is modeled as
	\begin{align}
		\bH=\bD\bG\bP.\label{EquivalentChannel}
	\end{align}
	\begin{remark}
		Herein, we assume statistical CSI, which means that the correlation matrices are available. These correlation matrices can be obtained based on the techniques \cite{Liang2001,Chen2010} after estimating the channel between the transmitter and the receiver according to the methods proposed in \cite{Liu2020a,Hu2021}.
	\end{remark}

	To this end, the received signal $\by \in \mathbb{C}^{N_{r}}$ is given by 
	\begin{align}
		\by=\bH\bs +\bn,
	\end{align}
	where 	$\bs \in \mathbb{C}^{N_{t}}$ is the transmitted signal with $\EE\{|s_{i}|^{2}\}=1, i=1,\ldots,N_{t}$, i.e., unit average transmit power, and $\bn\in \mathbb{C}^{N_{r}}\sim \mathcal{CN}(\b0, \sigma^{2}\Id_{N_{r}})$ the additive white 	Gaussian noise  with variance $\sigma^{2} $.
	\subsection{\textcolor{black}{Comparison between SIM and RIS in Holographic MIMO}}
	\textcolor{black}{The key differences between the new holographic MIMO architecture implemented using SIMs and the traditional RIS-based holographic MIMO lie in their system architecture, signal processing capabilities, and electromagnetic characteristics.}
	
	\textcolor{black}{	From a deployment perspective, traditional RIS-based holographic MIMO systems utilize a single passive metasurface placed within the propagation environment (e.g., mounted on a wall). This metasurface reflects or re-radiates incoming signals toward the receiver by applying programmable phase shifts. In contrast, the proposed SIM-based holographic MIMO employs two metasurfaces, one integrated at the transmitter and one at the receiver, both of which actively participate in the signal transmission and reception process, rather than passively reflecting signals.
	}
	
	\textcolor{black}{	Structurally and mathematically, RIS-based systems are typically modeled using diagonal phase shift matrices, as the RIS comprises a single-layer passive structure with independent phase control per element. On the other hand, SIMs consist of multilayer metasurfaces with significant inter-element electromagnetic interactions. This leads to a non-diagonal transformation matrix that captures mutual coupling and layered wave interactions, enabling richer and more complex propagation behavior.
	}
	
	\textcolor{black}{	In terms of functionality and control, RISs provide passive beamforming with limited control, typically phase-only or amplitude-phase optimization, and lack intrinsic signal processing capabilities. SIMs, however, offer more advanced control over the electromagnetic wavefronts, potentially supporting scattering, modulation, and spatial encoding. This enables a more comprehensive manipulation of wave propagation through engineered electromagnetic responses.
	}	
	
	\textcolor{black}{	Regarding implementation and application, RIS-based holographic MIMO is primarily suited for passive wavefront shaping and is commonly used for coverage enhancement or line-of-sight assistance. In contrast, SIM-based holographic MIMO enables direct integration with the transceiver’s RF chain and supports transmission- and reception-side holographic beamforming. This approach enhances spatial resolution, increases channel capacity, and opens the possibility for full-aperture communication.}
	
	\section{\textcolor{black}{EMI Analysis Using Large RMT}}\label{mutualI}
	To support the presentation of the main results, we provide certain preliminary results based on large random matrix theory, which is efficient for the analysis of large MIMO systems without requiring lengthy Monte Carlo simulations \cite{Couillet2011a, Couillet2011, Papazafeiropoulos2015a}. These results rely on the  assumptions below.
	
	\begin{assumption}\label{as1}
		$0<\lim\inf\limits_{N_{t}\ge 1} \frac{N_{t}}{M}\le  \frac{N_{t}}{M}\le \lim\sup\limits_{N_{t}\ge 1} \frac{N_{t}}{M}< \infty$, 	$0<\lim\inf\limits_{N\ge 1} \frac{N}{N_{r}}\le  \frac{N}{N_{r}}\le \lim\sup\limits_{N\ge 1} \frac{N}{N_{r}}< \infty$, $0<\lim\inf\limits_{N_{t}\ge 1} \frac{N_{t}}{N_{R}}\le  \frac{N_{t}}{N_{r}}\le \lim\sup\limits_{N_{t}\ge 1} \frac{N_{t}}{N_{r}}< \infty$.
	\end{assumption}
	\begin{assumption}\label{as2}
		$\lim\sup\limits_{N_{R}\ge 1}\|\bR_{\mathrm{R}}\|< \infty$, $
		\lim\sup\limits_{N_{t}\ge 1}\|\bR_{\mathrm{T}}\|< \infty$.
	\end{assumption}
	\begin{assumption}\label{as3}
		$\lim\inf\limits_{M\ge 1}\frac{1}{M}\tr(\bR_{\mathrm{T}})>0$, 	$\lim\inf\limits_{N\ge 1}\frac{1}{N}\tr(\bR_{\mathrm{R}})>0$.
	\end{assumption}
	Regarding the assumptions, Assumption \ref{as1}	defines the asymptotic regime, where $N_{t}, N_{r}$, $ M$, and $ N$ go to infinity with the same ratios. Assumptions \ref{as2} and \ref{as3} consider a restriction on the rank of the correlation matrices, i.e., the corresponding ranks  increase with the number of the SIM elements.
	\subsection{Mutual Information}
	The MI of the  SIM-aided point-to-point HMIMO communication system is given by
	\begin{align}
		C(\rho)=\log\det(\Id_{N_{r}}+\rho\bH\bH^{\H}),\label{MI}
	\end{align}
	where $\rho=\frac{P}{N_{t}\sigma^{2}}$ with $P$ expressing the  total transmit power.
	
	By plugging \eqref{EquivalentChannel} and \eqref{channel} into \eqref{MI}, we obtain
	\begin{align}
		C(\rho)=\log\det(\Id_{N_{r}}+\rho\bD\bR^{1/2}_{\mathrm{R}}\tilde{\bG}\bR^{1/2}_{\mathrm{T}}\bP\bP^{\H}\bR^{1/2}_{\mathrm{T}}\tilde{\bG}^{\H}\bR^{1/2}_{\mathrm{R}}\bD^{\H}).\label{MI1}
	\end{align}
	
	From \eqref{MI1}, it can be seen that the equivalent channel matrix can be written as
	\begin{align}
		\bar{\bH}=\bD\bR^{1/2}_{\mathrm{R}}\tilde{\bG}\bR^{1/2}_{\mathrm{T}}\bP\label{equivalent channel}.
	\end{align} 
	
	A tight approximation of the MI as $N_{r}\to \infty$ follows.
	\begin{theorem}[Couillet2011a]\label{Th1}
		Given the Assumptions \ref{as1}-\ref{as3}, the average throughput $\EE\{C(\rho)\}$  obeys to
		\begin{align}
			\frac{1}{N_{r}}\EE\{C(\rho)\}\xrightarrow{N_{r}\to \infty}\frac{1}{N_{r}}\bar{C}(\rho). 
		\end{align}
		In our case, where $\tilde{\bG}$ is Gaussian, it holds that
		\begin{align}
			\EE\{C(\rho)\}\xrightarrow{N_{r}\to \infty}\bar{C}(\rho), 
		\end{align}
		where $\bar{C}(\rho)$ is given by
		\begin{align}
			\bar{C}(\rho)&=\log \det (\Id_{N_{t}}+
			\frac{N_{r}}{N_{t}} e(-\rho)\bP^{\H}\bR_{\mathrm{T}}\bP)\nn\\
			&+\log\det(\Id_{N_{r}}+\delta(-\rho)\bD\bR_{\mathrm{R}}\bD^{
				\H})-\rho \delta (-\rho) e(-\rho)\label{MeanMI}.
		\end{align}
		
		Note that $(\delta(\rho),e(\rho)) $ forms  the unique positive solution of the following
		system of equations
		\begin{align}	
			\!	\!\!\!e(\rho)=\frac{1}{N_{r}}\tr \left(\bD\bR_{\mathrm{R}}\bD^{
				\H}\bQ_{r}\right), 
			\delta(\rho)=\frac{1}{N_{t}}\tr\left(\bP^{\H}\bR_{\mathrm{T}}\bP \bQ_{t}\right),\label{solution}
		\end{align}
		where
		\begin{align}
			\bQ_{r}&=\left(-\rho\Big[\Id_{N_{r}}+	\delta(\rho)\bD\bR_{\mathrm{R}}\bD^{
				\H}\Big]\right)^{-1},\nn\\ 
			\bQ_{t}&=\left(-\rho\bigg[\Id_{N_{t}}+		\frac{N_{r}}{N_{t}}	e(\rho)\bP^{\H}\bR_{\mathrm{T}}\bP\bigg]\right)^{-1}.
		\end{align}
	\end{theorem}
	
	\begin{remark}
		The tight approximation, described by Theorem \ref{Th1}, holds for any random matrix $\tilde{\bG}$, which is not necessarily Gaussian. Generally, $	\bar{C}(\rho)$ has to be normalized by $N_{r}$, but in the case of Gaussianity, we can omit this normalization.  The solution of \eqref{solution} depends on the phase shifts of the two SIMs through the matrices $\bPhi^{l}$ and $\bPsi^{k}$, respectively.
	\end{remark}
	
	The study of the fluctuation of the MI follows by presenting the corresponding central limit theorem (CLT).
	
	\subsubsection{The CLT for the MI}
	
	Regarding the Gaussianity of the MI, we present the following theorem, which is the CLT corresponding to the proposed system model.
	\begin{theorem}\label{th2}
		With Assumptions \ref{as1}-\ref{as3} satisfied, the CLT for SIM-assisted  HMIMO systems is written as
		\begin{align}
			\frac{{C}(\rho)-\bar{C}(\rho)}{\sqrt{V(\rho)}}\xrightarrow[N_{r}\to \infty]{\mathcal{D}}\mathcal{N}(0,1),\label{CLT}
		\end{align}
		where $ \bar{C}(\rho)$ is given by \eqref{MeanMI} while the variance $V(\rho)$ at the asymptotic limit is written as
		\begin{align}
			V(\rho)=-\log(1-\gamma\tilde{\gamma})
		\end{align}
		with  $\gamma$ and $\tilde{\gamma}$ given in Appendix~\ref{th2proof}.
	\end{theorem}
	\begin{proof}
		Please see Appendix~\ref{th2proof}.	
	\end{proof}
	
	\subsection{Special Case}
	In this subsection, we will consider the optimal transmission, which assumes a limited number of data streams	$S$. Also, to reveal the impact of the sizes of the SIMs, we assume independent Rayleigh channels, which means that $\bR_{\mathrm{T}}=\Id_{M}$ and $\bR_{\mathrm{R}}=\Id_{N}$. In particular, for the optimal transmission, we apply the truncated singular value decomposition (SVD) method to $\bG$. To this end, we have $\bG=\bE\bLambda\bF^{\H}$, where $ \bLambda=\diag(\lambda_{1},\ldots, \lambda_{\min(M,N)})$, where $\lambda_{1}\ge\ldots,\ge \lambda_{\min(M,N)}$ are the singular values in non-increasing order. Moreover, by spreading the data streams in terms of a transmit precoder $\bP=\bF_{:,1:S} \in \mathbb{C}^{M\times S}$ and collecting the spatial
	signals using a receive combiner $\bD=\bE_{:,1:S} \in \mathbb{C}^{S\times N}$. In such case, we have $\bP\bP^{\H}=\Id_{M}$ and $\bD^{\H}\bD=\Id_{N}$ since they are unitary matrices.
	
	\begin{proposition}
		Given Assumption \ref{as1}, the average throughput $\EE\{C(\rho)\}$ of SIM-assisted MIMO	systems over independent channels with $ N_{r}=N_{t}$  is obtained as
		\begin{align}
			\bar{C}(\rho)&=M \log  (1+
			\frac{1}{\rho(1+\delta(\rho))} )+N\log(1+\delta(\rho))\nn\\
			&-\frac{ \delta (\rho)}{1+\delta(\rho)} \label{MeanMIspecial},
		\end{align}
		where   $\delta$ is given by
		\begin{align}
			\delta=\frac{-\rho-2+\sqrt{\rho^{2}+8\rho+4}}{2 \rho}
		\end{align}
		with $\delta>0$.
	\end{proposition}
	\begin{proof}
		The proof is straightforward by substituting $\bR_{\mathrm{T}}=\Id_{M}$, $\bR_{\mathrm{R}}=\Id_{N}$, $\bP\bP^{\H}=\Id_{M}$, and $\bD^{\H}\bD=\Id_{N}$, while the property $\det(\Id+\bA\bB)=\det(\Id+\bB\bA)$ is applied. Note that  $\delta$ is  obtained by the 	equation
		\begin{align}
			\rho\delta^{2}+(\rho+2)\delta-1=0.
		\end{align}
	\end{proof}
	
	As can be seen, $\bar{C}(\rho)$  now depends only on one parameter, which is $\delta$. In this case, $\delta$ is the solution of a quadratic equation, which is related to the SNR $\rho$. Also, the performance depends on $M$, $N$, which are the sizes of the surfaces of the two SIMs.
	
	\begin{proposition}
		Let Assumption \ref{as1} and $ N_{r}=N_{t}$, the CLT for the MI of SIM-assisted MIMO	systems over independent channels with  is obtained by \eqref{CLT},
		where   the variance is written as
		\begin{align}
			V(z)=-\log(1-\frac{1}{1+\rho f}\frac{1}{z+\tilde{f}}),
		\end{align}
		where $f= \frac{1}{1+\rho\tilde{f}}$, and $\tilde{f}=\frac{1}{1+\rho f}$.
	\end{proposition}
	\begin{proof}
		The proof is straightforward by substituting $\bR_{\mathrm{T}}=\Id_{M}$, $\bR_{\mathrm{R}}=\Id_{N}$, $\bP\bP^{\H}=\Id_{M}$, and $\bD^{\H}\bD=\Id_{N}$.
	\end{proof}
	
	\section{\textcolor{black}{Outage Analysis and Joint SIM Optimization}}\label{outageProb}
	The following theorem presents the approximated outage probability, which is obtained as
	\begin{align}
		P_{\mathrm{out}}(R)=\mathbb{P}({C}(\rho)<R).
	\end{align}
	\begin{theorem}
		Let a transmission rate $R$, the outage probability of SIM-assisted HMIMO systems can be approximated in closed form as
		\begin{align}
			P_{\mathrm{out}}(R)\approx \Xi \left(\frac{R-\bar{C}(\rho)}{\sqrt{V(\rho)}}\right)\label{outage}\!.
		\end{align}
	\end{theorem}
	\begin{proof}
		The proof is straightforward given the CLT provided by Theorem \ref{th2}.
	\end{proof}

	
	\subsubsection{Problem Formulation and Optimization}
	Based on statistical CSI, the formulation of the minimization of the outage probability is described as
	\begin{subequations}\label{eq:subeqns}
		\begin{align}
			(\mathcal{P})~~&\min_{\bphi_{l},\bpsi_{k}} 	\;			P_{\mathrm{out}}(R) \label{Maximization1} \\
			&\;\quad\;\;\;\;\;\;\;\!\!\!~\!		\bD=\bU^{1}\bPsi^{1}\bU^{2}\bPsi^{2}\cdots\bU^{K}\bPsi^{K},
			\label{Maximization4} \\
			&~	\mathrm{\textcolor{black}{s.t.}}\;\;\;\;\;\;\!\!~\!	\bP=\bPhi^{L}\bW^{L}\cdots\bPhi^{2}\bW^{2}\bPhi^{1}\bW^{1},
			\label{Maximization3} \\
			&\;\quad\;\;\;\;\;\;\;\!\!\!~\!		\bPsi^{k}=\diag(\psi^{k}_{1}, \dots, \psi^{k}_{N}), k \in \mathcal{K},
			\label{Maximization6} \\
			&\;\quad\;\;\;\;\;\;\;\!\!\!~\!		\bPhi^{l}=\diag(\phi^{l}_{1}, \dots, \phi^{l}_{M}), l \in \mathcal{L},
			\label{Maximization5} \\
			&	\;\quad\;\;\;\;\;\;\;\!\!\!~\!		|\psi^{k}_{n}|=1, n \in \mathcal{N}, k \in \mathcal{K}	\label{Maximization8},\\
			&\;\quad\;\;\;\;\;\;\;\!\!\!~\!		|	\phi^{l}_{m}|=1, m \in \mathcal{M}, l \in \mathcal{L}.	\label{Maximization7} \\
		\end{align}
	\end{subequations}
	
	Problem $ (\mathcal{P}) $ is nonconvex because the  objective function is  neither   convex nor concave regarding $\bPhi^{l}$ and $\bPsi^{k}$. Also, we have non-convex constant modulus constraints. In this section, given $\rho$, both $\bar{C}(\rho)$  and $	V(\rho)$ are denoted in terms of the phase shifts of the SIMs. Thus, we have $\bar{C}(\bphi_{l},\bpsi_{k})$  and $	V(\bphi_{l},\bpsi_{k})$. The objective function is written as
	\begin{align}
		f(\bphi_{l},\bpsi_{k}) =\Xi \left(\frac{R-\bar{C}(\bphi_{l},\bpsi_{k})}{\sqrt{V(\bphi_{l},\bpsi_{k})}}\right)\!.
	\end{align}
	
	\subsection{Proposed Algorithm}
	Despite the coupling of the  parameters $e$ and $\delta$ with $\phi^{l}_{m}$ and $\psi^{k}_{n}$, the partial derivatives of $f(\bphi_{l},\bpsi_{k})$  can be derived in closed forms. As a result, the solution to the optimization problem could be achieved by applying the gradient descent approach. Instead of using alternating optimization (AO) to optimize the phase shifts separately in an alternating way, we optimize the transmitter and receiver SIMs  simultaneously. This approach is proposed because the convergence of the AO method may require many iterations that actually increase with the size of the surfaces \cite{Perovic2021}.
	
	The proposed algorithm
is outlined below. Note that the parameters $ \mu_{n}^{i} >0$ for $ i=1,2 $ define  the step of the move  towards the gradient of the objective function. Also, we denote
	\begin{align}
			\Psi_{k}&=\{\bpsi_{k}\in \mathbb{C}^{N \times 1}: |\psi^{k}_{i}|=1, i=1,\ldots, N\},\\
				\Phi_{l}&=\{\bphi_{l}\in \mathbb{C}^{M \times 1}: |\phi^{l}_{i}|=1, i=1,\ldots, M\}.
	\end{align}
	
	In Algorithm \ref{Algoa1}, $ P_{\Phi_{l}}(\cdot) $ and $ P_{\Psi_{k}}(\cdot) $ express the projections onto   $ \Phi_{l} $ and $ \Psi_{k} $, respectively. The reason for the projections is that the newly computed points may be found outside of the feasible set.
	
	\begin{algorithm}[th]
		\caption{Projected Gradient Ascent Method\label{Algoa1}}
		\begin{algorithmic}[1]
			\STATE Input: $\bphi_{l}^{0},\bpsi_{k}^{0},\mu_{n}^{\textcolor{black}{i}}>0$ \textcolor{black}{for $ i=1,2 $}.
			\STATE \textbf{for} $ n=1,2,\ldots \textbf{do} $
			\STATE ~~~~~$\bphi_{l}^{n+1}=P_{\Phi_{l}}(\bphi_{l}^{n}+\mu_{n}^{\textcolor{black}{1}}\nabla_{\bphi_{l}}f(\bphi_{l}^{n},\bpsi_{k}^{n}))$
			\STATE ~~~~~$\bpsi_{k}^{n+1}=P_{\Psi_{k}}(\bpsi_{k}^{n}+\mu_{n}^{\textcolor{black}{2}}\nabla_{\bpsi_{k}}f(\bphi_{l}^{n},\bpsi_{k}^{n}))$
			\STATE \textbf{end for}
		\end{algorithmic}
	\end{algorithm} 	
	%
	
	
	The presentation of the gradients 
	$\nabla_{\bphi_{l}}f(\bphi_{l},\bpsi_{k}) $ and $ \nabla_{\bpsi_{k}}f(\bphi_{l},\bpsi_{k}) $ follows.
	
	\textcolor{black}{From a computational perspective, both the proposed gradient-based method and AO involve similar per-iteration complexity, primarily scaling as  $ \mathcal{O}(L M^{2}+K N^{2})$ 	 due to the computation of gradient terms across all layers and meta-atoms. However, the proposed gradient method jointly updates all SIM phases, whereas AO optimizes them sequentially. In particular, AO requires sequential updates of each SIM, which is less efficient in the presence of strong inter-layer coupling and non-diagonal transformations characteristic of SIM-based architectures. As the number of SIM layers $L$ and $K$ increases, the coupling intensifies, causing AO to require significantly more iterations to reach convergence, as confirmed in Fig. 4. In contrast, the joint gradient descent method updates all variables simultaneously, capturing global interactions and thus converging faster. As a result, the proposed method converges in fewer iterations, as demonstrated in Fig. 4, making it more scalable for systems with large numbers of elements or layers.}
	
	\subsection{Derivation of Gradients}
	The  gradients of $f(\bphi_{l},\bpsi_{k}) $ are given by the following proposition in closed forms.
	\begin{proposition}\label{lemmaGradient}
		The gradients of $f(\bphi_{l},\bpsi_{k}) $ with respect to $ i=\bphi_{l}^{*}, \bpsi_{k}^{*}$  are provided by
		\begin{align}
			\nabla_{i}	f(\bphi_{l},\bpsi_{k}) \!=\!\frac{-	\nabla_{i}	\bar{C}(\bphi_{l},\bpsi_{k})V\!-\!\frac{1}{2}(R\!-\!\bar{C})	\nabla_{i}	V(\bphi_{l},\bpsi_{k})}{V^{\frac{3}{2}}},
		\end{align}
		where $\nabla_{i}	\bar{C}(\bphi_{l},\bpsi_{k})$ and $\nabla_{i}	V(\bphi_{l},\bpsi_{k})$ are obtained by Lemmas \ref{lemmaGradient1} and \ref{lemmaGradient2}.
	\end{proposition}
	\begin{proof}
		Please see Appendix~\ref{lem1}.	
	\end{proof}
	
	\begin{lemma}\label{lemmaGradient1}
		The gradients of  $\bar{C} $ with respect to $ i=\bphi_{l}^{*}, \bpsi_{k}^{*}$  are
		\begin{align}
			&	\nabla_{\bphi_{l}}	\bar{C}(\bphi_{l},\bpsi_{k})=\frac{N_{r}}{N_{t}}(e \diag(\bA_{l}^{\H}(\bK\bP^{\H}\bR_{\mathrm{T}}))\nn\\
			&+\frac{1}{N_{t}}(1+\rho e)\diag(\bA_{l}^{\H}(\bQ_{t}\bP^{\H}\bR_{\mathrm{T}}))\nn\\
			&+(\tr({\bK} \bar{\bR}_{\mathrm{T}})
			-\frac{\rho}{N_{t}}(1+e)\tr((\bar{\bR}_{\mathrm{T}}\bQ_{t})^{2})-\rho \delta)\frac{\partial e}{\partial{(\bphi^{l*})}}\nn\\
			&-\frac{\rho e}{N_{t}}(1+\rho e)\diag(\bA_{l}^{\H}(\bQ_{t}\bar{\bR}_{\mathrm{T}}\bQ_{t}\bP^{\H}\bR_{\mathrm{T}})))
			,\label{gradient2}\\
			&	\nabla_{\bpsi_{k}}	\bar{C}(\bphi_{l},\bpsi_{k})=\frac{N_{r}}{N_{t}}(\diag(\bC_{k}^{\H}(\bQ_{r}\bD\bR_{R}))\nn\\
			&-\frac{N_{r}}{N_{t}}\rho\diag(\bC_{k}^{\H}(\bQ_{r}\bD {\bR}_{\mathrm{R}}\bD^{\H}\bQ_{r}\bar{\bR}_{\mathrm{R}}))\nn\\
			&-	\frac{N_{r}}{N_{t}^{2}}(1+N_{t}(e+\delta))\rho\tr((\bar{\bR}_{\mathrm{T}}\bQ_{t})^{2})\frac{\partial e}{\partial{(\bpsi^{k*})}}\nn\\
			&-\delta \diag(\bC_{k}^{\H}(\bB\bD\bR_{\mathrm{R}}))),\label{gradient3}
		\end{align}
		with \begin{align}
			\frac{\partial}{\partial{(\bphi^{l*})}}e&= -\frac{\rho}{N_{t} N_{r}}\tr((\bar{\bR}_{\mathrm{R}}\bQ_{r})^{2})(\diag(\bA_{l}^{\H}(\bQ_{t}\bP^{\H}\bR_{\mathrm{T}}))\nn\\
			&-\rho	\frac{N_{r}}{N_{t}}\tr((\bQ_{t}\bar{\bR}_{\mathrm{T}})^{2})   \frac{\partial e}{\partial{(\bphi^{l*})}}\nn\\
			& -\frac{\rho  e}{N_{t} }\diag(\bA_{l}^{\H}(\bQ_{t}\bar{\bR}_{\mathrm{T}} \bQ_{t}\bP^{\H}\bR_{\mathrm{T}})),
		\end{align}
		\begin{align}
			&\frac{\partial}{\partial{(\bpsi^{k*})}}e=\frac{1}{N_{r}}\diag(\bC_{k}^{\H}(\bQ_{r}\bD\bR_{\mathrm{R}}))\nn\\
			&-\frac{\rho \delta}{N_{r}}\diag(\bC_{k}^{\H}(\bQ_{r}\bar{\bR}_{\mathrm{R}}\bQ_{r}\bD\bR_{R}))\nn\\
			&-\frac{\rho }{N_{r}}\tr((\bar{\bR}_{\mathrm{R}}\bQ_{r})^{2})\tr((\bar{\bR}_{\mathrm{T}}\bQ_{t})^{2})	\frac{\partial}{\partial{(\bpsi^{k*})}}e,
		\end{align}
		where 
		\begin{align}
			&{\bK}=(\Id_{N_{t}}+
			\frac{N_{r}}{N_{t}} e(-\rho)\bar{\bR}_{\mathrm{T}})^{-1},\\
			&{\bB}=(\Id_{N_{r}}+\delta(-\rho)\bar{\bR}_{\mathrm{R}})^{-1},\\
			&\bA_{l}(\bX)=\bW^{l}\bPhi^{l-1}\bW^{l-1}\cdots \bPhi^{1}\bW^{1} \bX\bPhi^{L}\nn\\
			&\times\bW^{L}\cdots\bPhi^{l+1}\bW^{l+1},\\
			&	\bC_{k}(\bY)=\bU^{k}\bPsi^{k-1}\bU^{k-1}\cdots \bPsi^{1}\bU^{1}	\bY	\bPsi^{K} \nn\\
			&\times\bU^{K}\cdots \bPsi^{k+1}\bU^{k+1}.	
		\end{align}
	\end{lemma}
	\begin{proof}
		Please see Appendix~\ref{lem2}.	
	\end{proof}
	\begin{lemma}\label{lemmaGradient2}
		The gradients of  $V $ with respect to $ i=\bphi_{l}^{*}, \bpsi_{k}^{*}$  are
		\begin{align}
			\nabla_{i} V=\frac{\tilde{\gamma}\nabla_{i}\gamma+\gamma \nabla_{i}\tilde{\gamma}}{1-\gamma\tilde{\gamma}},
		\end{align}
		where
		\begin{align}
			&\nabla_{\bphi_{l}}\gamma=-\frac{2}{N_{t}}  \tr \left(\right.(\bar{\bR}_{\mathrm{R}}\bS )^{2}	\bS\bar{\bR}_{R}\bS\left.\right)\nn\\
			&\times (\frac{1}{N_{t}}\diag(\bA_{l}^{\H}(\tilde{\bS}\bP^{\H}\bR_{T}))+\frac{\rho}{N_{t}} \tr((\bar{\bR}_{T}\tilde{\bS})^{2})\nabla_{\bphi_{l}} \tilde{f}\nn\\
			&+\frac{\rho f}{N_{t}}\diag(\bA_{l}^{\H}(\bP^{\H}\bR_{T}\tilde{\bS}\bar{\bR}_{T}\tilde{\bS}))),\\
			&\nabla_{\bphi_{l}}\tilde{\gamma}=	\frac{2}{N_{t}}\diag(\bA_{l}^{\H}(\tilde{\bS}\bar{\bR}_{T}\tilde{\bS}\bP^{\H} \bR_{T}))\nn\\
			&-\frac{2 \rho}{N_{t}}\tr((\bar{\bR}_{T}	\tilde{\bS})^{2}      )	\nabla_{\bphi_{l}} f\nn\\
			&-\frac{2\rho}{N_{t}} \diag(\bA_{l}^{\H}(\tilde{\bS}\bar{\bR}_{T}\tilde{\bS}\bP^{\H}\bR_{T})),\\
			&\nabla_{\bpsi_{k}}\gamma=\frac{ 2}{N_{t}}\diag(\bC_{k}^{\H}(\bS\bR_{\mathrm{R}}))\nn\\
			&-\frac{ 2f \rho}{N_{t}}\tr((\bar{\bR}_{\mathrm{R}})^{2})\diag(\bC_{k}^{\H}(\bS\bD\bR_{\mathrm{R}}))\nn\\
			&-\frac{ 2 \rho}{N_{t}}\tr((\bar{\bR}_{\mathrm{R}})^{2}\bS\bar{\bR}_{\mathrm{R}}\bS)\nabla_{\bpsi_{k}}\tilde{f},\nn\\
			&\nabla_{\bpsi_{k}} \tilde{\gamma}=-\frac{ 2 \rho }{N_{t}}\tr(\bar{\bR}_{\mathrm{T}}\tilde{\bS}\bar{\bR}_{\mathrm{T}}\tilde{\bS}\bar{\bR}_{\mathrm{T}}\tilde{\bS})\nabla_{\bpsi_{k}}f
		\end{align}
	\end{lemma}
	with
	\begin{align}
		&	\nabla_{\bphi_{l}}f=-\frac{\rho}{N_{t}}\tr((\bar{\bR}_{R}\bS)^{2})\nabla_{\bphi_{l}}\tilde{f},\\
		&\nabla_{\bphi_{l}} \tilde{f}=-\frac{\rho}{N_{t}}\tr((\bar{\bR}_{T}\tilde{\bS})^{2})\nabla_{\bphi_{l}}f,\\
		&\nabla_{\bpsi_{k}}f=-\frac{\rho}{N_{t}}\tr((\bar{\bR}_{\mathrm{R}}{\bS})^{2})\nabla_{\bpsi_{k}}\tilde{f},\\
		&\nabla_{\bpsi_{k}}\tilde{f}=-\frac{\rho}{N_{t}}\tr((\bar{\bR}_{\mathrm{T}}\tilde{\bS})^{2})\nabla_{\bpsi_{k}}{f}.
	\end{align}
	\begin{proof}
		Please see Appendix~\ref{lem3}.	
	\end{proof}
	
\textcolor{black}{It is worth noting that although the proposed projected gradient descent algorithm solves a non-convex optimization problem and thus cannot guarantee convergence to a global minimum, it consistently achieves significantly better performance compared to the unoptimized baseline (i.e., random phase configuration). This performance enhancement is verified through numerical results below, which demonstrate lower outage probability and higher ergodic mutual information across various system configurations.}\

\begin{remark}
	`\textcolor{black}{In this work, we assume ideal and instantaneous control of SIM phase shifts. However, in practical systems with large metasurfaces, the delivery of high-resolution phase control information can introduce significant signaling overhead and latency. Recent approaches based on phase shift compression, such as attention-based encoding \cite{Yu2024}, offer promising solutions for reducing control overhead. Extending such techniques to SIM-based architectures, potentially through joint PSI compression across multiple layers, is a valuable direction for future work.	}
\end{remark}
	
	\section{\textcolor{black}{Finite-SNR Trade-offs in SIM-Assisted Systems}}\label{DMT}
	This section  addresses the finite-SNR DMT of SIM-assisted HMIMO systems. Note that initially, DMT was studied in \cite{Zheng2003} to evaluate the trade-off between multiplexing and diversity gain. This DMT was described as the large-SNR DMT. In this work, we focus on the  DMT for finite SNR\cite{Narasimhan2006,Loyka2010}, which sheds light on the low and moderate SNR regions.
	\begin{definition}[\cite{Loyka2010}]\label{def1}
		The multiplexing gain is obtained as
		\begin{align}
			w=\frac{q R}{\EE\{C\}} \approx\frac{q R}{\bar{C}},
		\end{align}
		where $R $ is the data rate and $q=\min(M, N)$
	\end{definition}

	\begin{theorem}\label{Theorem3}
		The  finite-SNR DMT of  SIM-assisted HMIMO systems is given by
		\begin{align}
			d(w,\rho)=\frac{z(d-w)}{w \sqrt{2\pi}}\frac{\exp\left(-\frac{(d-w)^{2}G^{2}(z)}{2 w^{2}}\right)G'(z)}{\Xi\left(\frac{(d-w)G(z)}{ w}\right)},
		\end{align}
		where $G(z)=\frac{\bar{C}}{\sqrt{V(z)}}$ with  $G'(z)=\frac{\mathrm{d}G(z)}{\mathrm{d} z}$ obtained
		as
		\begin{align}
			G'(z)=	\frac{\bar{C}'(z)V(z)-\frac{1}{2}\bar{C}(z)V'(z)}{V^{\frac{3}{2}}(z)}.
		\end{align}
		The derivatives of $\bar{C}(z)$ and ${V}(z)$ with respect to $z$ are provided in Appendix \ref{Theorem3proof}. 
	\end{theorem}
	\begin{proof}
		Please see Appendix \ref{Theorem3proof}.
	\end{proof}
	
	\subsection{\textcolor{black}{Key Factors Driving SIM Superiority over Single-Layer RIS}}
\textcolor{black}{	The performance gains of the proposed SIM-based holographic MIMO system over conventional RIS-assisted architectures stem from several fundamental differences in design and electromagnetic behavior.}
	
	\textcolor{black}{First, unlike traditional RIS, which typically consists of a single passive layer with independent phase control per element, a SIM is a multilayer metasurface. This structural enhancement enables more complex electromagnetic transformations, resulting in richer scattering behavior and improved wavefront manipulation.
	}
	
	\textcolor{black}{Second, the mathematical representation of SIMs involves a non-diagonal transformation matrix due to strong inter-element coupling and mutual interactions across layers. This is in contrast to RISs, which are typically modeled by diagonal phase-shift matrices, thereby limiting their ability to perform coordinated spatial processing. The non-diagonal structure in SIMs allows for more advanced electromagnetic interference (EMI) shaping and spatial resolution.
	}
	
	\textcolor{black}{Third, the proposed system deploys SIMs at both the transmitter and the receiver, providing bidirectional wavefront control. This dual-SIM configuration significantly enhances the system’s spatial degrees of freedom compared to the single, passively deployed RIS typically used in conventional architectures.
	}
	
	\textcolor{black}{Finally, SIMs offer broader electromagnetic functionality, enabling simultaneous control over amplitude, phase, and angular response. This increased control diversity facilitates more effective energy focusing, reduced path loss, and more robust signal delivery, particularly in complex propagation environments.
	}
	
	\textcolor{black}{These combined advantages explain the superior performance of the SIM-based system in terms of EMI suppression and outage probability, as validated by the theoretical analysis and simulation results presented in this work.}
	\section{Numerical Results}\label{Numerical}
	In this section, we elaborate on the performance of large SIM-assisted MIMO systems in terms of EMI and outage probability through the demonstration of  both theoretical results and Monte Carlo (MC) simulations. According to the simulation setup,  both SIMs, i.e., at the transmitter and the receiver are parallel to the $ x-z $ plane and centered  along the $ y-$axis at a height $ H_{\mathrm{BS}}=5~\mathrm{m} $. Also, the implementation of both SIMs follows the same guidelines. In particular, their parameter values are defined based on \cite{An2023,Papazafeiropoulos2024a}. Hence, the area of each meta-atom  is $ (\lambda/2)^{2}  $ while the spacing  between adjacent meta-atoms  is assumed to be $ \lambda/2 $. The spacing among the metasurfaces is  $ d_{\mathrm{SIM}}= T_{\mathrm{SIM}}/L$, where the thickness of the SIMs is $ T_{\mathrm{SIM}}= R_{\mathrm{SIM}}=5 \lambda $. The path-loss exponent is $b=2.5$, and   the distance between the two SIMs is  $d=200\mathrm{m}$.    The correlation matrices $\bR_{\mathrm{T}}, \bR_{\mathrm{R}}$ are obtained according to \eqref{t} and \eqref{c}. \textcolor{black}{Moreover, in our implementation, the step sizes $\mu_{n_i}=0.01$ 		are set to constant values tuned empirically to ensure convergence. While fixed step sizes provide stable performance in our experiments, adaptive step size strategies could be employed to further improve convergence speed and robustness, especially in more complex or time-varying environments.}  The    system bandwidth and carrier frequency  are $ 20~\mathrm{MHz} $ and $ 2~\mathrm{GHz} $, respectively. The  transmit power available is  $\rho=20\mathrm{dBm}$, and the receiver sensitivity is set to $\sigma^{2}=-110\mathrm{dBm}$ Furthermore, we assume  $ N_{t}= N_{r}=32$, $ M=N=200 $, and  $ K=L=4 $ unless otherwise stated.

	In Fig. \ref{fig2}, we show the EMI versus the number  of surfaces $ K $ while changing  $L$, which is the  number of surfaces of the transmitter SIM. We notice that the  EMI increases with both  $ K $ and $L $. However, an increase further approximately to 5 surfaces does not provide an extra increase to the EMI. In the same figure, we have depicted the performance  corresponding to instantaneous CSI in the case of $K=4$ surfaces. This case performs better than the statistical CSI case since the optimization takes place at each  coherence interval instead of at each of several coherence intervals. However, in the case of statistical CSI, a significant overhead is saved due to lower complexity and less frequent optimization. Moreover, MC simulations verify the analytical results since the corresponding lines coincide.	
	
	In Fig. \ref{fig3}, we depict the EMI versus the number of   elements  $ M $ of each metasurface of the transmitter  while varying the number of elements  $N$ of each metasurface of the receiver. The EMI increases with both $M $ and $N$. Moreover, we have simulated the EMI for random phase shifts without any optimization when $N=200 $ elements, which shows that the optimization is  necessary for better performance since the EMI is higher when it has been optimized.		
	
	In Fig. \ref{fig4}, we illustrate the	convergence of the gradient ascent algorithm, where the two SIMs are optimized simultaneously against the AO method. In the latter method, the two SIMs are optimized independently in an alternating way. The algorithm finishes  when the iterations are greater than $100$ or the difference  between the two last iterations is less than $10^{-5}$. It is shown that the proposed algorithm saturates faster to the optimum EMI, while the AO method is slower due to higher overhead. Furthermore, we consider two cases, i.e., $ M = N = 100$  and $ M = N = 200$ for $L=5$. We observe that more elements  require more iterations to converge since the complexity increases with the number of elements.

	Fig. \ref{fig5} depicts the outage probability versus the transmit power. As can be seen, the outage probability decreases with the transmit power and the transmit rate $R$. Specifically, we have plotted three lines, which correspond to $R=2,3,4$, respectively. In parallel, we have plotted MC simulations, which corroborate the analytical results since the corresponding lines coincide.
	\begin{figure}%
		\centering
		\includegraphics[width=0.9\linewidth]{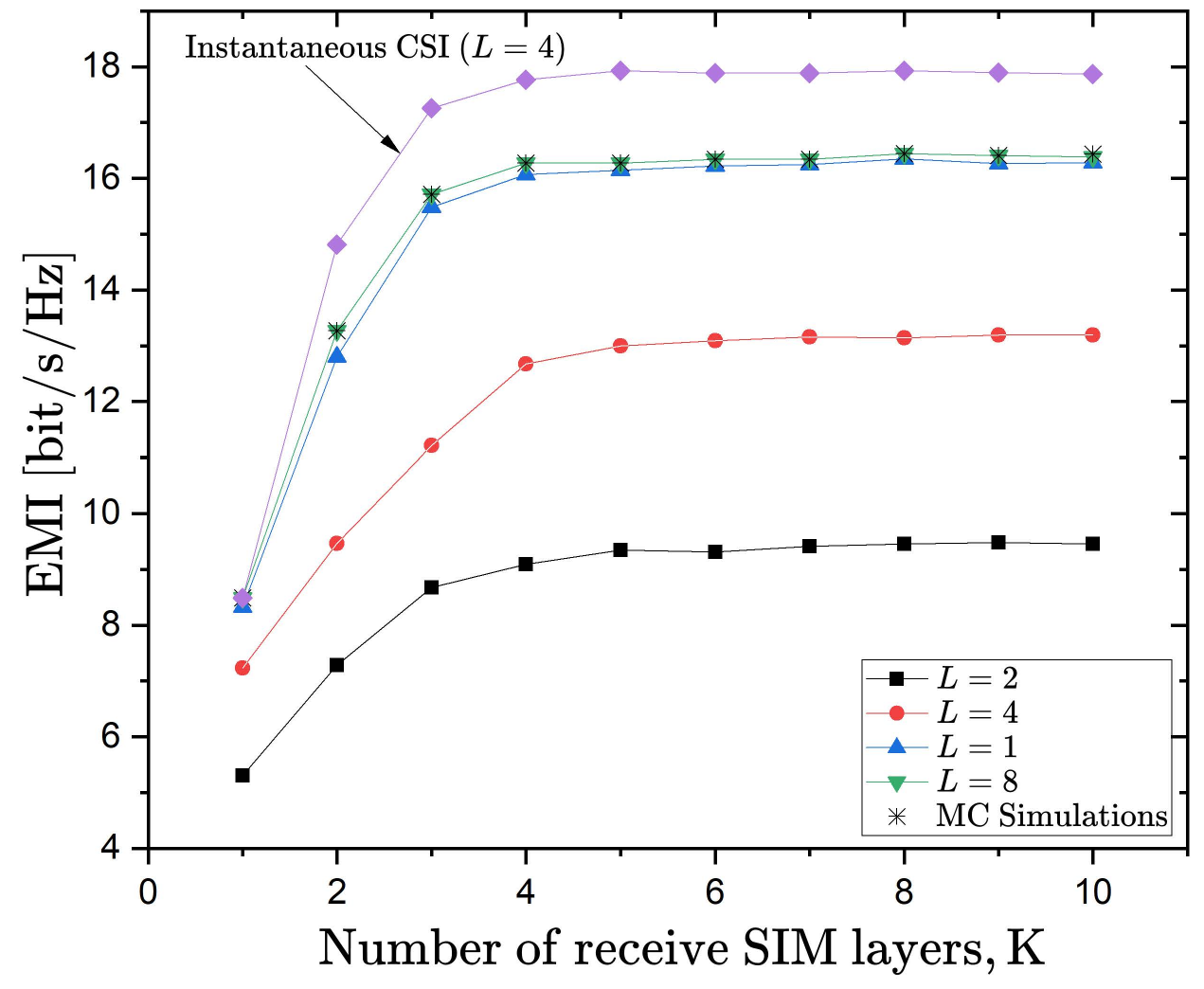}
		\caption{EMI with respect to the number of SIM layers of the receiver SIM.}
		\label{fig2}
	\end{figure}
	
	Fig. \ref{fig6} shows the outage probability with respect to the transmit power for a varying number of metasurfaces $L$ at the receiver (dotted lines) and a number of elements per metasurface (solid lines). For a given rate threshold, an increase of both parameters decreases the outage probability, i.e., the performance is improved. \textcolor{black}{From Fig.  \ref{fig6}, we observe that increasing the number of SIM elements (per layer) enhances both diversity and multiplexing performance. This is attributed to the improved spatial resolution and effective aperture gain. Furthermore, as the number of SIM layers increases, the system experiences greater angular diversity and richer scattering, which results in a more favorable DMT curve, especially in the high-reliability (low multiplexing) regime. However, due to inherent mutual coupling and practical implementation limits, the DMT performance eventually saturates. These trends are further illustrated in Fig. 9, which shows how the finite-SNR DMT varies with different SIM configurations.}
	
	Fig. \ref{fig7} shows the outage probability  with respect to the rate threshold  for a varying number of metasurfaces $L$ at the receiver (solid lines) and a number of elements per metasurface (dotted lines). It is shown that under a given transmit power, the outage probability approaches $1$ as the rate threshold increases.
	
	Fig. \ref{fig8}  depicts the diversity and multiplexing trade-off for different SNRs. It is shown that the diversity gain decreases with increasing multiplexing gain, while it increases with SNR. The figure also validates the expression of the finite SNR DMT provided by Theorem \ref{Theorem3}.
	\begin{figure}%
		\centering
		\includegraphics[width=0.9\linewidth]{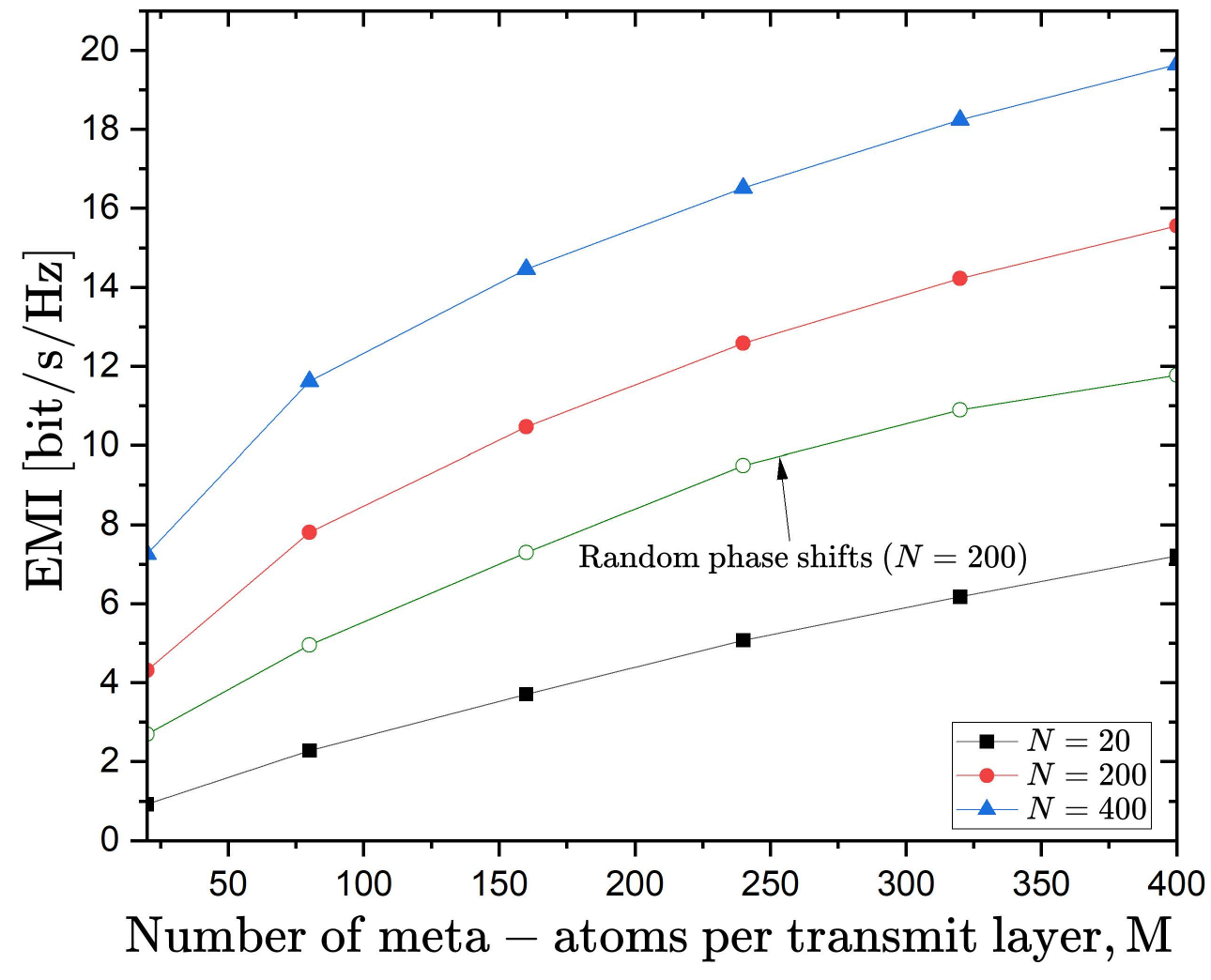}
		\caption{EMI   with respect to the number of metasurface elements of the transmitter SIM.}
		\label{fig3}
	\end{figure}
	
	\begin{figure}%
		\centering
		\includegraphics[width=0.9\linewidth]{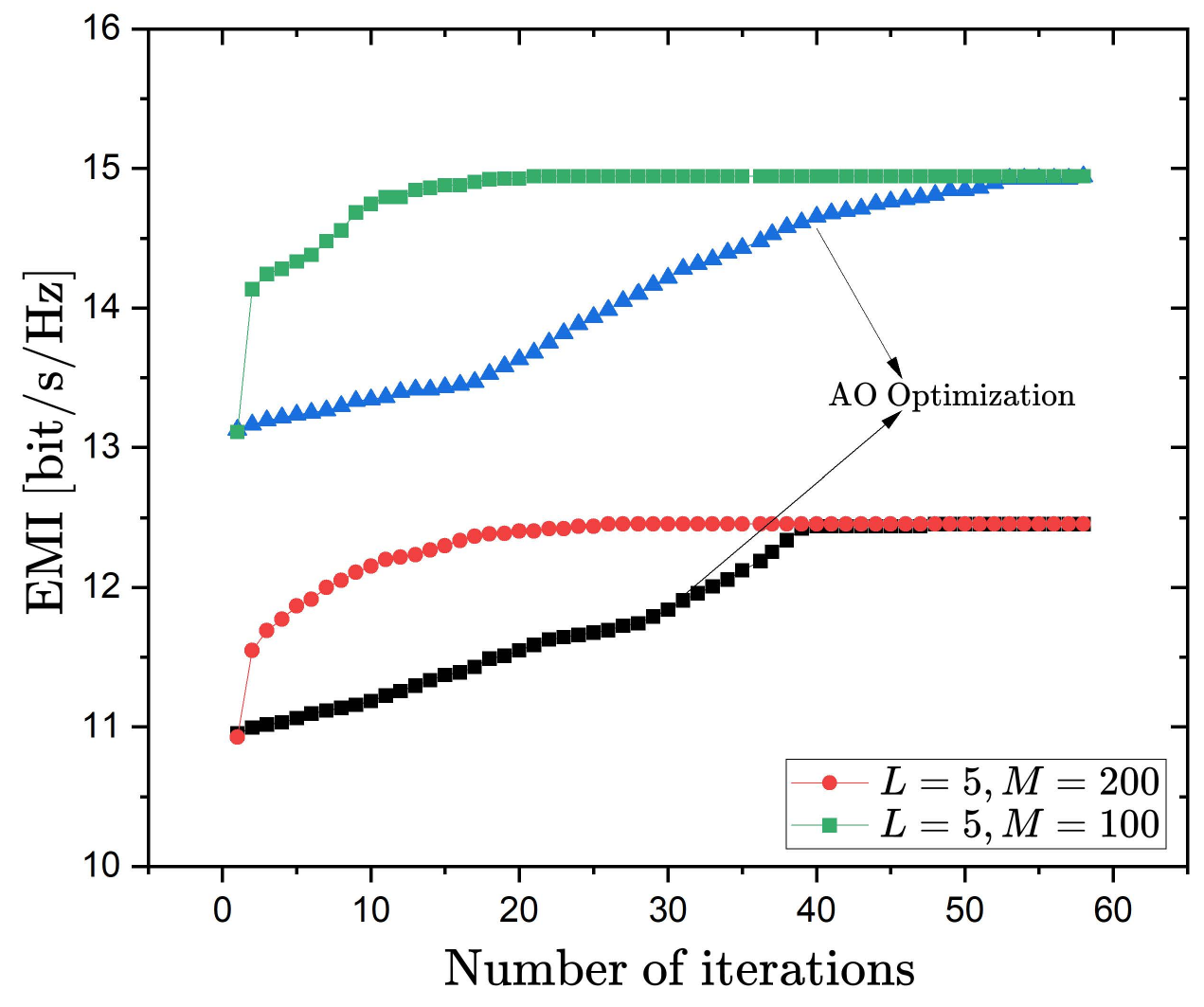}
		\caption{EMI  with respect to the number of iterations.}
		\label{fig4}
	\end{figure}

	\begin{figure}%
		\centering
		\includegraphics[width=0.95\linewidth]{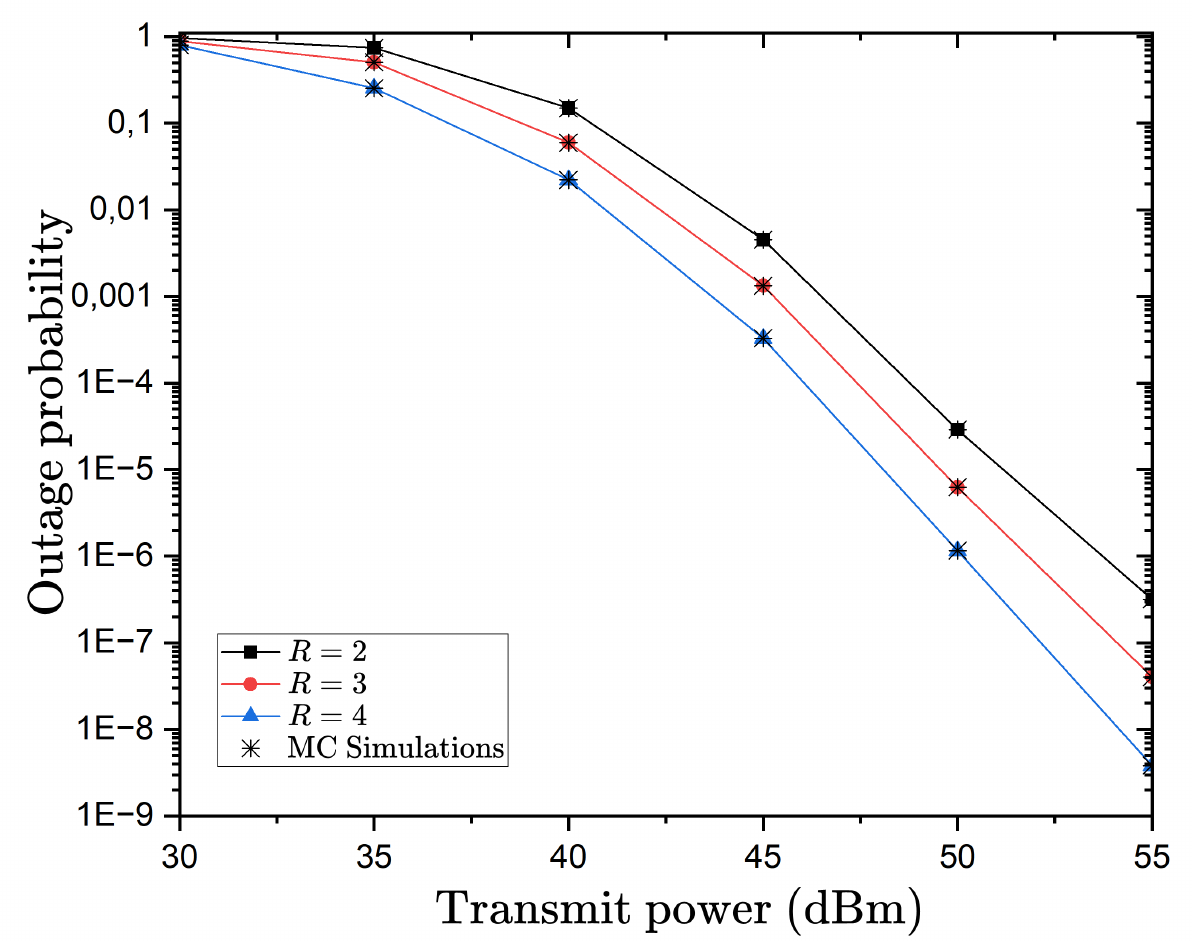}
		\caption{Outage probability  with respect to the transmit power.}
		\label{fig5}
	\end{figure}	
	
	\begin{figure}%
		\centering
		\includegraphics[width=0.95\linewidth]{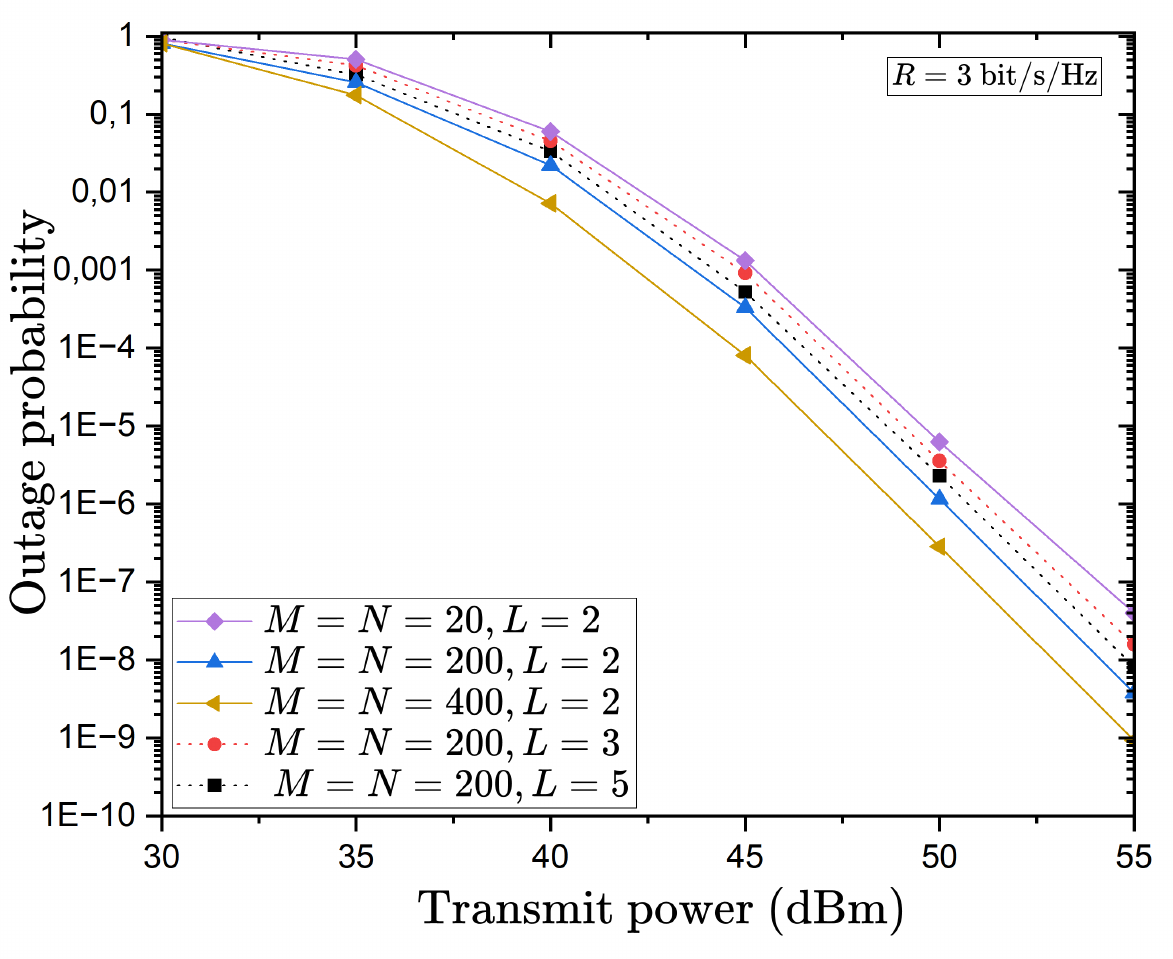}
		\caption{Outage probability  with respect to the transmit power.}
		\label{fig6}
	\end{figure}	
	
	\begin{figure}%
		\centering
		\includegraphics[width=0.95\linewidth]{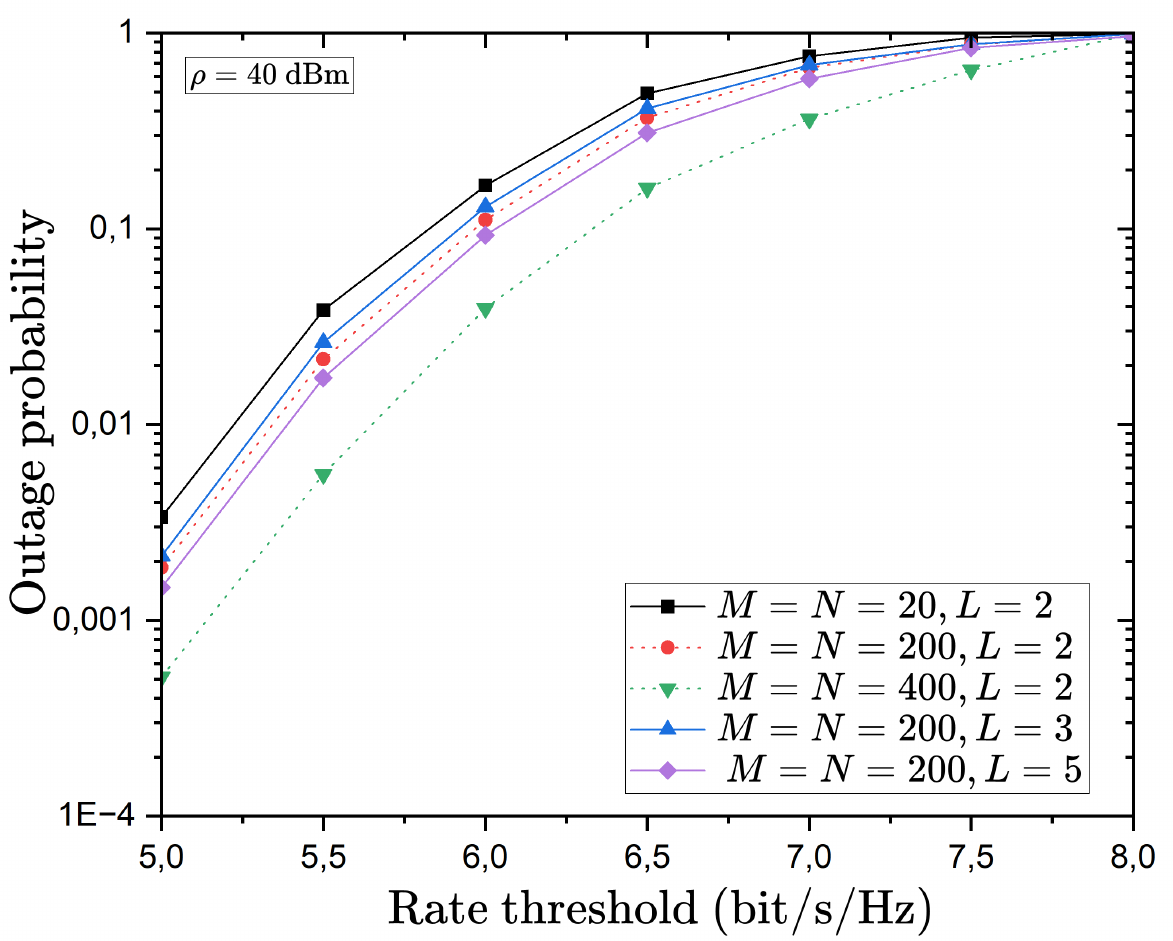}
		\caption{Outage probability with respect to the rate threshold.}
		\label{fig7}
	\end{figure}	
	
	\begin{figure}%
		\centering
	\includegraphics[width=0.95\linewidth]{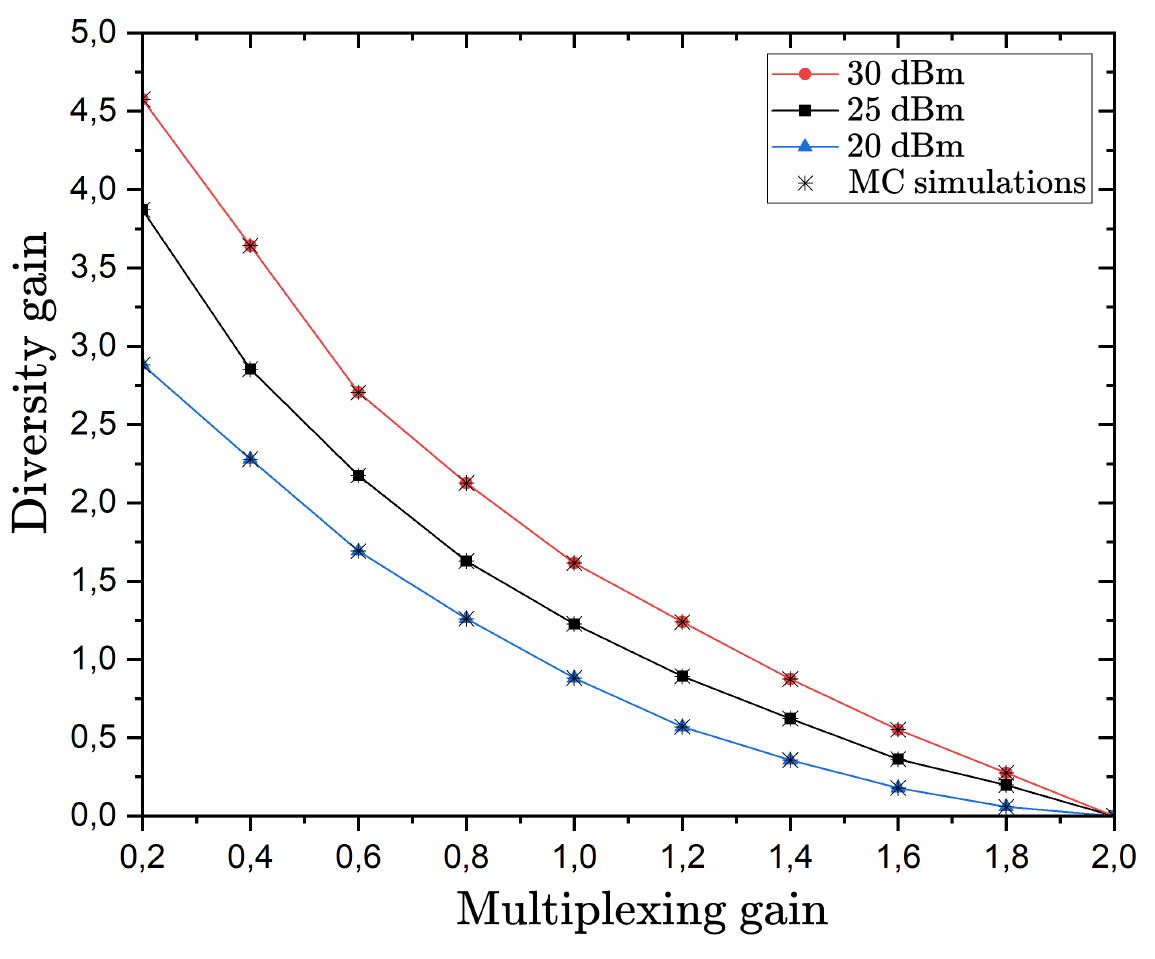}
		\caption{Diversity and multiplexing trade-off varying power.}
		\label{fig8}
	\end{figure}

	\textcolor{black}{Fig. \ref{fig9} shows how the finite-SNR DMT behavior evolves with different SIM configurations (varying number of layers $L$ and elements $N$). The key trends visible in the figure are that  more SIM layers (e.g., $L = 6$) lead to higher diversity gains, especially in the low multiplexing regime. Also, a larger aperture size (e.g., $M = 128$) improves both diversity and multiplexing trade-offs, pushing it upward. The diversity gain decreases with increasing multiplexing gain, but this trade-off is mitigated as $L$ and $M$ increase.}
	
		\begin{figure}%
		\centering
		\includegraphics[width=0.95\linewidth]{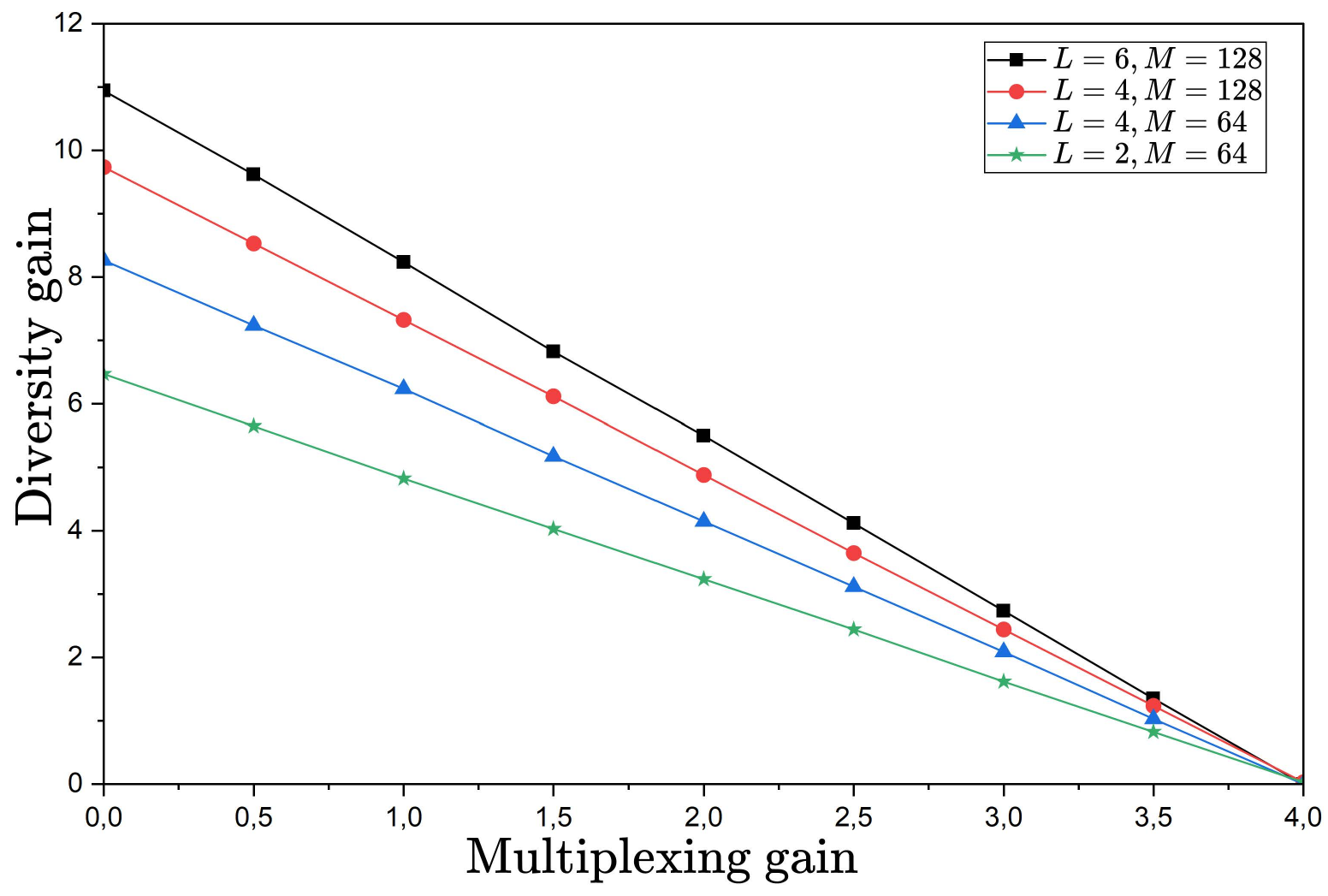}
		\caption{Diversity and multiplexing trade-off varying $L$ and $M$.}
		\label{fig9}
	\end{figure}	
	
	\section{Conclusion} \label{Conclusion} 
	In this paper, we derived the EMI, the outage probability, and finite SNR DMT for large SIM-assisted MIMO systems based on statistical CSI. In particular, by exploiting RMT results, the EMI was derived. Next, use of the CLT resulted in an approximation of the outage probability in closed form. Furthermore, the phase shifts were optimized by proposing a gradient descent algorithm for the minimization of the outage probability. Moreover, a closed-form finite SNR DMT was obtained to provide a trade-off between the throughput and the outage probability. Finally, analytical and simulation results were depicted and showed that the sizes of the SIMs improve the performance.	
	
	\begin{appendices}
		\section{Proof of Theorem~\ref{th2}}\label{th2proof}	
		Taking into account \cite[Th. 2]{Hachem2008}, we conclude the convergence of ${C}(\rho)-\bar{C}(\rho)$ to a zero-mean Gaussian process with asymptotic variance given by
		\begin{align}
			V(\rho)=\mathrm{Var}({C}(\rho))\xrightarrow[N_{t}\to \infty]{\mathcal{P}}-\log(1-\gamma\tilde{\gamma}),
		\end{align}
		where $\gamma=\frac{1}{N_{t}}\tr(\bar{\bR}_{\mathrm{R}}\bS)^{2}$, $\tilde{\gamma}=\frac{1}{N_{t}}\tr(\bar{\bR}_{\mathrm{T}}\tilde{\bS})^{2} $, $\bS=(\Id_{N_{r}}+\rho\tilde{f}\bar{\bR}_{\mathrm{R}})^{-1}$, $\tilde{\bS}=(\Id_{N_{t}}+\rho{f}\bar{\bR}_{\mathrm{T}})^{-1}$
		with $(f,\tilde{f})$ being the unique positive solution of the 
		system of equations $f=\frac{1}{N_{t}}\tr(\bar{\bR}_{\mathrm{R}}\bS)$ and $\tilde{f}=\frac{1}{N_{t}}\tr(\bar{\bR}_{R}\tilde{\bS})$. 
		
		\section{Proof of Proposition~\ref{lemmaGradient}}\label{lem1}
		The derivation of the gradient starts by focusing on its differential. Specifically, we have
		\begin{align}
			d(f(\bphi_{l},\bpsi_{k}))=\frac{\exp\left(-\frac{W^{2}(\bphi_{l},\bpsi_{k})}{2}\right) d(W(\bphi_{l},\bpsi_{k}))}{\sqrt{2 \pi}},
		\end{align}
		where $W(\bphi_{l},\bpsi_{k})=\frac{R-\bar{C}(\bphi_{l},\bpsi_{k})}{\sqrt{V(\bphi_{l},\bpsi_{k})}}$,
		\begin{align}
			d(W) =\frac{-d(\bar{C})V-\frac{1}{2}(R-\bar{C})d(V)}{V^{\frac{3}{2}}}.
		\end{align}
		By taking the corresponding derivative, we obtain the gradient expression, where $\nabla_{i}	\bar{C}(\bphi_{l},\bpsi_{k})$ and $\nabla_{i}	V(\bphi_{l},\bpsi_{k})$ are obtained by Lemmas \ref{lemmaGradient1} and \ref{lemmaGradient2}.
		\section{Proof of Lemma~\ref{lemmaGradient1}}\label{lem2}
		From \eqref{MeanMI}, we have 
		\begin{align}
			&d(\bar{C})	=\frac{N_{r}}{N_{t}}\tr({\bK} (d(e)	 \bar{\bR}_{\mathrm{T}}+e (d( \bP^{\H})\bR_{\mathrm{T}}\bP+ \bP^{\H}\bR_{\mathrm{T}}d(\bP))))
			\nn\\
			&+\tr({\bB}(d(\delta)\bar{\bR}_{\mathrm{R}}+\delta(d(\bD)\bR_{\mathrm{R}}\bD^{
				\H}+\bD\bR_{\mathrm{R}}d(\bD^{
				\H}))))\nn\\
			&	-\rho (d(\delta) e+\delta d (e))=d(I_{1})+d(I_{2})-d(I_{3}),\label{differential}
		\end{align}
		where we have used that $ d(\det(\bX))=\det(\bX)\tr(\bX^{-1}d(\bX)) $. Also, we have denoted ${\bK}=(\Id_{N_{t}}+
		\frac{N_{r}}{N_{t}} e(-\rho)\bar{\bR}_{\mathrm{T}})^{-1}$ and ${\bB}=(\Id_{N_{r}}+\delta(-\rho)\bar{\bR}_{\mathrm{R}})^{-1}$ with $ \bar{\bR}_{\mathrm{T}}=\bP^{\H}\bR_{\mathrm{T}}\bP$ and $\bar{\bR}_{\mathrm{R}}=\bD\bR_{\mathrm{R}}\bD^{
			\H}$.
		
		For the differential $d(\bP)$, we have
		\begin{align}
			d(\bP)&=\bPhi^{L}\bW^{L}\cdots\bPhi^{l+1}\bW^{l+1}d(\bPhi^{l})\bW^{l}\bPhi^{l-1}\nn\\
			&\times\bW^{l-1}\cdots\bPhi^{1}\bW^{1}. \label{differentialPhi4}
		\end{align}
		A similar expression to \eqref{differentialPhi4} holds for the differential $d(\bD)$. The differential of 	$e$		 is obtained as
		\begin{align}
			d(e)\!=\!\frac{1}{N_{r}}\!\tr \left(d(\bD)\bR_{\mathrm{R}}\bD^{
				\H}\bQ_{r}\!+\!\bD\bR_{\mathrm{R}}d(\bD^{
				\H})\bQ_{r}\!+\!\bar{\bR}_{\mathrm{R}}d(\bQ_{r})\right)\!.\label{de}
		\end{align}
		In  particular, $d(\bQ_{r})$ concerns the differential of an inverse matrix. Generally, for an inverse matrix $\bT$, its differential is obtained as \cite[Eq. 40]{Petersen2012}  $ d(\bT)=-\bT d(\bT^{-1})\bT $. Hence, we have
		\begin{align}
			d(\bQ_{r})=-\bQ_{r} d(\bQ_{r}^{-1})\bQ_{r},\label{dqr}
		\end{align}		
		where 
		\begin{align}
			\!\!\!d(\bQ_{r}^{-1})=-\rho(	d(\delta)\bar{\bR}_{\mathrm{R}}+\delta d(\bD)\bR_{\mathrm{R}}\bD^{
				\H}+\delta\bD\bR_{\mathrm{R}}(d\bD^{
				\H})).\label{dqr1}
		\end{align}			
		Insertion of \eqref{dqr1}	into \eqref{dqr}	gives	$	d(\bQ_{r})$.
	
	Similarly, the  differential of 	$\delta$		 is derived as
	\begin{align}
		d(\delta)&=\frac{1}{N_{t}}\tr\left(\right. d(\bP^{\H})\bR_{\mathrm{T}}\bP \bQ_{t}+\bP^{\H}\bR_{\mathrm{T}}d(\bP) \bQ_{t}\nn\\
		&+\bar{\bR}_{\mathrm{T}} d( \bQ_{t})\left.\right),
	\end{align}			
	where 
	\begin{align}
		d(\bQ_{t})=-\bQ_{r} d(\bQ_{t}^{-1})\bQ_{t},\label{dqt}
	\end{align}			
	with
	\begin{align}
		d(\bQ_{r}^{-1})&=-\rho	\frac{N_{r}}{N_{t}}(	d(e)\bar{\bR}_{\mathrm{T}}+	e d(\bP^{\H})\bR_{\mathrm{T}}\bP+	\nn\\
		&+e\bP^{\H}\bR_{\mathrm{T}}d(\bP)).
	\end{align}		
	
	Thus,   $d(\delta) $ becomes
	\begin{align}
		d(\delta)&=\frac{1}{N_{t}}\tr\left(\!\right. d(\bP^{\H})\bR_{\mathrm{T}}\bP \bQ_{t}+\bP^{\H}\bR_{\mathrm{T}}d(\bP) \bQ_{t}\nn\\
		&+\bar{\bR}_{\mathrm{T}} d(\bQ_{t})
		\!\left.\right).\label{dd2}
	\end{align}

	Substitution of $	d(\bP)$  into the first trace of \eqref{differential} gives
	\begin{align}
		&d(I_{1})=\tr({\bK} (d(e)	\bar{\bR}_{\mathrm{T}}+e (d( \bP^{\H})\bR_{\mathrm{T}}\bP+ \bP^{\H}\bR_{\mathrm{T}}d(\bP))))
		\nn\\
		&=d(e)\tr(	 {\bK}\bar{\bR}_{\mathrm{T}})+e (\bA_{l}^{\H}(\bK\bP^{\H}\bR_{\mathrm{T}}) d((\bPhi^{l})^{\H})\nn\\
		&+ 	\bA_{l}(\bK\bP^{\H}\bR_{\mathrm{T}})   d(\bPhi^{l}))),\label{firstterm}
	\end{align}			
	where 
	\begin{align}
		\bA_{l}(\bK\bP^{\H}\bR_{\mathrm{T}})&=\bW^{l}\bPhi^{l-1}\bW^{l-1}\cdots \bPhi^{1}\bW^{1} \bK\bP^{\H}\bR_{\mathrm{T}}\bPhi^{L}\nn\\
		&\times\bW^{L}\cdots\bPhi^{l+1}\bW^{l+1}.
	\end{align}			
	
	The derivative of the first term of \eqref{differential} can be written as
	\begin{align}
		\frac{\partial I_{1}}{\partial{(\bphi^{l*})}}	=\tr({\bK} \bar{\bR}_{\mathrm{T}})
		\frac{\partial e}{\partial{(\bphi^{l*})}}+e \diag(\bA_{l}^{\H}(\bK\bP^{\H}\bR_{\mathrm{T}})),\label{partiale}
	\end{align}			
	where we have used that $ \tr\left(\bA\bB\right)=(\left(\diag\left(\mathbf{A}\right)\right)^{\T}d(\diag(\bB))$ for any matrices $ \bA,\bB $ with $ \bB $ being a diagonal matrix. 
	
	From \eqref{de}, the partial derivative $\frac{\partial e}{\partial{(\bphi^{l*})}}$ in \eqref{partiale} is obtained as
	\begin{align}
		&\frac{\partial e}{\partial{(\bphi^{l*})}}=\frac{1}{N_{r}}\tr(\bar{\bR}_{\mathrm{R}}\frac{\partial \bQ_{r}}{\partial{(\bphi^{l*})}})\\
		&=-\frac{\rho}{N_{r}}\tr((\bar{\bR}_{\mathrm{R}}\bQ_{r})^{2})\frac{\partial \delta}{\partial{(\bphi^{l*})}}
		)\\
		&=-\frac{\rho}{N_{t} N_{r}}\tr((\bar{\bR}_{\mathrm{R}}\bQ_{r})^{2})(\diag(\bA_{l}^{\H}(\bQ_{t}\bP^{\H}\bR_{\mathrm{T}}))+ \frac{\partial \bQ_{t}}{\partial{(\bphi^{l*})}})\\
		&=-\frac{\rho}{N_{t} N_{r}}\tr((\bar{\bR}_{\mathrm{R}}\bQ_{r})^{2})(\diag(\bA_{l}^{\H}(\bQ_{t}\bP^{\H}\bR_{\mathrm{T}}))\nn\\
		&-\rho	\frac{N_{r}}{N_{t}}\tr((\bQ_{t}\bar{\bR}_{\mathrm{T}})^{2})   \frac{\partial e}{\partial{(\bphi^{l*})}}\nn\\
		& -\frac{\rho N_{r} e}{N_{t} N_{r}}\diag(\bA_{l}^{\H}(\bQ_{t}\bar{\bR}_{\mathrm{T}} \bQ_{t}\bP^{\H}\bR_{\mathrm{T}}))\label{de4}
		).
	\end{align}

	In a similar way,   calculation of the derivative with respect to  $\bphi^{l*}$ of the second trace of \eqref{differential} yields
	\begin{align}
		\frac{\partial}{\partial{(\bphi^{l*})}}	I_{2}=\tr({\bB}\bar{\bR}_{\mathrm{R}})\frac{\partial}{\partial{(\bphi^{l*})}}\delta,
	\end{align}
	where
	\begin{align}
		\frac{\partial}{\partial{(\bphi^{l*})}}\delta&=\frac{1}{N_{t}}\diag(\bA_{l}^{\H}(\bQ_{t}\bP^{\H}\bR_{\mathrm{T}}))\nn\\
		&-\frac{\rho}{N_{t}}\tr((\bar{\bR}_{\mathrm{T}}\bQ_{t})^{2}\frac{\partial}{\partial{(\bphi^{l*})}}e)\nn\\
		&-\frac{\rho e}{N_{t}}\diag(\bA_{l}^{\H}(\bQ_{t}\bar{\bR}_{\mathrm{T}}\bQ_{t}\bP^{\H}\bR_{\mathrm{T}})) \label{de5}.
	\end{align}

	The  derivative of the last term $I_{3}$ is straightforward since \eqref{de4} and \eqref{de5} have already been derived. Thus, $	\nabla_{\bphi_{l}}\bar{C}$ is obtained. In a similar way, $	\nabla_{\bpsi_{k}}\bar{C}$ is derived.
	
	\section{Proof of Lemma~\ref{lemmaGradient2}}\label{lem3}
	Regarding the variance, its differential is given by
	\begin{align}
		d(V)=\frac{d(\gamma)\tilde{\gamma}+\gamma d(\tilde{\gamma})}{1-\gamma\tilde{\gamma}}.
	\end{align}
	We have
	\begin{align}
		&d(\gamma)=\frac{2}{N_{t}}\tr \left(\right.\bar{\bR}_{\mathrm{R}}\bS\left(\right.(d(\bD)\bR_{\mathrm{R}}\bD^{
			\H}\bS+\bD\bR_{\mathrm{R}}d(\bD^{
			\H})\bS\nn\\
		&+\bar{\bR}_{\mathrm{R}}d(\bS)\left.\right)\left.\right),\\
		&d(\tilde{\gamma})=\frac{2}{N_{t}}\tr\left(\right.\bar{\bR}_{\mathrm{T}}\tilde{\bS}\left(\right.\bP^{\H}\bR_{\mathrm{T}}d(\bP)\tilde{\bS}\nn\\
		&+d(\bP^{\H})\bR_{\mathrm{T}}\bP\tilde{\bS}+\bar{\bR}_{\mathrm{T}}d(\tilde{\bS})\left.\right)\left.\right).
	\end{align}
	The differentials of $\bD$  and $\bP$ have already been derived. Hence, we focus on the differentials of $\bS$  and $\tilde{\bS}$. We have
	\begin{align}
		&\!\!d(\bS)\!=\!\rho \bS( d(\tilde{f})\bD\bR_{\mathrm{R}}\bD^{
			\H}\!+\!\tilde{f} d(\bD)\bR_{\mathrm{R}}\bD^{
			\H}\!+\! \bD\bR_{\mathrm{R}}d(\bD^{
			\H}))\bS,\\
		&\!\!d(\tilde{\bS})\!=\!-\rho\tilde{\bS}( d({f} )\bar{\bR}_{\mathrm{T}}\!+\!{f}d(\bP^{\H})\bR_{\mathrm{T}}\bP\!+\! \bP^{\H}\bR_{\mathrm{T}})d(\bP))\tilde{\bS}
	\end{align}
	with
	\begin{align}
		&d(f)=\frac{1}{N_{t}}\tr(d(\bD)\bR_{\mathrm{R}}\bD^{
			\H}\bS+\bD\bR_{\mathrm{R}}d(\bD^{
			\H})\bS+\bar{\bR}_{\mathrm{R}}d(\bS)),\\
		&d(\tilde{f})=\frac{1}{N_{t}}\tr\left(\!\!\right.d(\bP^{\H})\bR_{\mathrm{T}}\bP\tilde{\bS}+\bP^{\H} \bR_{\mathrm{T}}d(\bR_{\mathrm{T}})\tilde{\bS}+\bar{\bR}_{\mathrm{T}}d(\tilde{\bS})\!\!\left.\right)\!.	 
	\end{align}
	By starting with the derivative of $\gamma$, we have
	\begin{align}
		&\frac{\partial}{\partial{(\bphi^{l*})}}\gamma=\frac{2}{N_{t}}  \tr \left(\right.\!\!(\bar{\bR}_{\mathrm{R}}\bS )^{2}	\frac{\partial}{\partial{(\bphi^{l*})}}\bS\!\!\left.\right),
	\end{align}
	where 
	\begin{align}
		\frac{\partial}{\partial{(\bphi^{l*})}}\bS=-\rho\bS\bar{\bR}_{R}\bS\frac{\partial}{\partial{(\bphi^{l*})}}\tilde{f}.
	\end{align}
	Next, we derive  $\frac{\partial}{\partial{(\bphi^{l*})}}\gamma$ after some algebraic manipulations. Specifically, we have
	\begin{align}
		&\frac{\partial}{\partial{(\bphi^{l*})}}\gamma=-\frac{2}{N_{t}}  \tr \left(\right.(\bar{\bR}_{\mathrm{R}}\bS )^{2}	\bS\bar{\bR}_{R}\bS\left.\right)\nn\\
		&\times (\frac{1}{N_{t}}\diag(\bA_{l}^{\H}(\tilde{\bS}\bP^{\H}\bR_{T}))+\frac{\rho}{N_{t}} \tr((\bar{\bR}_{T}\tilde{\bS})^{2})\frac{\partial \tilde{f}}{\partial{(\bphi^{l*})}}\nn\\
		&+\frac{\rho \tilde{f}}{N_{t}}\diag(\bA_{l}^{\H}(\bP^{\H}\bR_{T}\tilde{\bS}\bar{\bR}_{T}\tilde{\bS}))).\end{align}
	where
	\begin{align}
		\frac{\partial \tilde{f}}{\partial{(\bphi^{l*})}}=-\frac{\rho}{N_{t}}\tr((\bar{\bR}_{T}\tilde{\bS})^{2})\frac{\partial {f}}{\partial{(\bphi^{l*})}}.
	\end{align}
	In the case of $\tilde{\gamma}$, we have
	\begin{align}
		d(\tilde{\gamma})&=\frac{2}{N_{t}}\tr\left(\right.\!\!\bar{\bR}_{T}\tilde{\bS}\left(\right.d(\bP^{\H})\bR_{T}\bP\tilde{\bS}+\bP^{\H}\bR_{T}d(\bP)\tilde{\bS}\nn\\
		&+\bar{\bR}_{T} d(\tilde{\bS})\!\!\left.\right)\!\left.\negthinspace\right).
	\end{align}
	Its derivative is written as
	\begin{align}
		\frac{\partial}{\partial{(\bphi^{l*})}}\tilde{\gamma}&=\frac{2}{N_{t}}\diag(\bA_{l}^{\H}(\tilde{\bS}\bar{\bR}_{T}\tilde{\bS}\bP^{\H} \bR_{T}))\nn\\
		&+\frac{2}{N_{t}}\tr(\bar{\bR}_{T}	\frac{\partial \tilde{\bS}}{\partial{(\bphi^{l*})}})\nn\\
		&=\frac{2}{N_{t}}\diag(\bA_{l}^{\H}(\tilde{\bS}\bar{\bR}_{T}\tilde{\bS}\bP^{\H} \bR_{T}))\nn\\
		&-\frac{2 \rho}{N_{t}}\tr((\bar{\bR}_{T}	\tilde{\bS})^{2}      )	\frac{\partial f}{\partial{(\bphi^{l*})}}\nn\\
		&-\frac{2\rho}{N_{t}}\tr(\bR_{T}\bP\tilde{\bS}\bar{\bR}_{T}\tilde{\bS}	\frac{\partial \bP^{\H}}{\partial{(\bphi^{l*})}})\nn\\
		&=\frac{2}{N_{t}}\diag(\bA_{l}^{\H}(\tilde{\bS}\bar{\bR}_{T}\tilde{\bS}\bP^{\H} \bR_{T}))\nn\\
		&-\frac{2 \rho}{N_{t}}\tr((\bar{\bR}_{T}	\tilde{\bS})^{2}      )	\frac{\partial f}{\partial{(\bphi^{l*})}}\nn\\
		&-\frac{2\rho}{N_{t}} \diag(\bA_{l}^{\H}(\tilde{\bS}\bar{\bR}_{T}\tilde{\bS}\bP^{\H}\bR_{T})),
	\end{align}
	where
	\begin{align}
		\frac{\partial f}{\partial{(\bphi^{l*})}}=-\frac{\rho}{N_{t}}\tr((\bar{\bR}_{R}\bS)^{2})\frac{\partial \tilde{f}}{\partial{(\bphi^{l*})}}.
	\end{align}
	Hence, we have obtained the derivatives $	\frac{\partial}{\partial{(\bphi^{l*})}}{\gamma}$ and $ 	\frac{\partial}{\partial{(\bphi^{l*})}}\tilde{\gamma}$. Similarly, we derive $	\frac{\partial}{\partial{(\bpsi^{k*})}}{\gamma}$ and $ 	\frac{\partial}{\partial{(\bpsi^{k*})}}\tilde{\gamma}$. We omit the details to avoid repetition.
	\section{Proof of Theorem~\ref{Theorem3}}\label{Theorem3proof}
	The definition of finite-SNR DMT, provided by Definition \eqref{def1}, gives
	\begin{align}
		d(w,\rho)&=\frac{\partial \log\left(	P_{\mathrm{out}}(\frac{(d-w)G(z)}{w})\right)}{\partial{\log \rho}}\nn\\
		&=\frac{z(d-w)}{w \sqrt{2\pi}}\frac{\exp\left(-\frac{(d-w)^{2}G^{2}(z)}{2 w^{2}}\right)G'(z)}{\Xi\left(\frac{(d-w)G(z)}{ w}\right)},
	\end{align}
	where the computation of $ G'(z)$ requires the derivatives of $\bar{C}(z)$ and ${V}(z)$ with respect to $z$, which are obtained  by calculating $\bar{C}'(z)$ and ${V}'(z)$ after applying the chain rule. We have
	\begin{align}
		\bar{C}'(\rho)&=\frac{N_{r}}{N_{t}} e'(-\rho)\tr (\bK
		\bP^{\H}\bR_{\mathrm{T}}\bP)\nn\\
		&+\delta'(-\rho)\tr(\bB\bD\bR_{\mathrm{R}}\bD^{
			\H})+\delta (-\rho) e(-\rho)\nn\\
		&+\rho(\delta' (-\rho) e(-\rho)+\delta(-\rho) e'(-\rho)),
	\end{align}
	where $e'(-\rho)$ and $\delta'(-\rho)$ are obtained from \eqref{solution} 
	by the following system of equations.
	\begin{align}
		e'(-\rho)&=-\frac{1}{N_{r}}(\tr(\bar{\bR}_{R}\bQ_{r}^{2})+\rho\delta'(-\rho))\tr (\bar{\bR}_{R}\bQ_{r})^{2}\nn\\
		&+\frac{1}{N_{r}}\delta(-\rho)\tr ((\bar{\bR}_{R}\bQ_{r})^{2})\\
		\delta'(-\rho)	&=-\frac{1}{N_{t}}(\tr(\bar{\bR}_{T}\bQ_{t}^{2})+	\frac{N_{r}\rho}{N_{t}}	e'(-\rho))\tr(\left(\bar{\bR}_{T} \bQ_{t}\right)^{2})\nn\\
		&+	\frac{N_{r}}{N_{t}}	e(-\rho))\tr(\left(\bar{\bR}_{T} \bQ_{t}\right))^{2}.
	\end{align}
	Regarding the derivative of the variance, we have
	\begin{align}
		V'(\rho)=\frac{\gamma'\tilde{\gamma}+\gamma\tilde{\gamma}'}{1-\gamma\tilde{\gamma}},
	\end{align}
	where 
	\begin{align}
		\gamma'&=\frac{2}{N_{t}}\tr(\bar{\bR}_{R}\bS\bar{\bR}_{R}{\bS}'),\\
		\tilde{\gamma}'&=\frac{2}{N_{t}}\tr(\bar{\bR}_{\mathrm{T}}\tilde{\bS}\bar{\bR}_{\mathrm{T}}\tilde{\bS}').
	\end{align}
	with
	\textcolor{black}{\begin{align}
		\bS'&=-\bS(\tilde{f}\bar{\bR}_{R}+\rho\tilde{f}'\bar{\bR}_{R})\bS,\\
		\tilde{\bS}'&=-\tilde{\bS}({f}\bar{\bR}_{T}+\rho{f}'\bar{\bR}_{T})\tilde{\bS}.
	\end{align}	}
	We have $(f',\tilde{f}')$ forms the unique positive solution of the following
	system of equations 
	\begin{align}
		f'&=-\frac{1}{N_{t}}\tr((\tilde{\bR}_{\mathrm{R}}{\bS})^{2})(\tilde{f}+\rho \tilde{f}'),\\
		\tilde{f}'&=-\frac{1}{N_{t}}\tr((\tilde{\bR}_{\mathrm{T}}\tilde{\bS})^{2})(f+\rho f').
	\end{align}
	
\end{appendices}

\bibliographystyle{IEEEtran}

\bibliography{IEEEabrv,bibl}
\end{document}